\documentclass[final,leqno,onefignum,onetabnum]{siamltex1213}

\usepackage{amssymb}
\usepackage{amsmath}
\usepackage{epsfig}
\usepackage{ifpdf,color}
\usepackage{graphicx}
\usepackage{subfigure}
\usepackage{enumerate}
\usepackage{tabularx}
\usepackage{epstopdf}
\newtheorem{claim}{Claim}[section]
\newtheorem{remark}[theorem]{Remark}

 \def\RR{{\mathbb R}}
 \def\P{{\mathcal P}}
 \def\aP{{\mathfrak P}}
 \def\PP{{\mathfrak P}}
 \def\R{{\mathcal R}}
 \def\S{{\mathcal S}}
 \def\aS{{\mathfrak S}}
\def\Q{{\mathcal Q}}
 
 \def\ZZ{{\mathbb Z}}

 \def\F{{\mathcal F}}
 \def\H{{\mathcal H}}

\def\r{\ensuremath{r}}

 \newcommand{\dunion}{\dot{\cup}}
\newcommand{\ratio}{\ensuremath{6}}
\newcommand{\defeq}{\stackrel{\textup{def}}{=}}

\title{Approximating Minimum Cost Connectivity Orientation and
  Augmentation}
\author{
Mohit Singh\thanks{
       {H. Milton Stewart School of Industrial and Systems Engineering,}
       {Georgia Institute of Technology, Atlanta, USA.}
       {(mohitsinghr@gmail.com)}}
 \and
L\'aszl\'o A.\ V\'egh\thanks{
       {Dept.\ of Mathematics,}
       {London School of Economics, London, UK.}
      {(l.vegh@lse.ac.uk)}
}}

\begin{document}

\maketitle

\begin{abstract}
We investigate problems addressing combined connectivity augmentation
and orientations settings.
We give a polynomial-time 6-approximation algorithm for finding a
minimum cost subgraph of an undirected graph $G$ that admits an orientation
covering a nonnegative crossing $G$-supermodular demand function, as
defined by Frank \cite{frank80}. An important example is
  $(k,\ell)$-edge-connectivity, a common generalization of global and
  rooted edge-connectivity.

Our algorithm is based on a non-standard application of the iterative
rounding method. We observe that the  standard linear program with cut
constraints is not amenable and use an
alternative linear program with partition and co-partition
constraints instead. The proof requires a new type of
uncrossing technique on partitions and co-partitions.

We also consider the problem setting when the cost of an
edge can be different for the two possible orientations.
The problem becomes substantially more difficult already for the simpler
requirement of $k$-edge-connectivity.
Khanna, Naor, and Shepherd~\cite{Khanna05}  showed that the
integrality gap of the natural linear program is at most $4$ when $k=1$ and
conjectured that it is constant for all fixed $k$. We disprove this
conjecture by showing an $\Omega(|V|)$ integrality gap even when $k=2$.
\end{abstract}



\pagestyle{myheadings}
\thispagestyle{plain}

\section{Introduction}

\footnote{A preliminary version of the paper appeared in Proceedings of the Twenty-Fifth Annual ACM-SIAM Symposium on Discrete Algorithms, SODA 2014.}
Connectivity augmentation and orientation are two fundamental classes of
problems for graph connectivity. In connectivity augmentation,
we wish to add a minimum cost set of new edges to a graph to satisfy
certain connectivity requirements, for example, $k$-edge-connectivity.
This problem can be solved by a minimum cost spanning tree algorithm
for $k=1$, however, becomes NP-complete for every fixed value $k\ge
2$. There is a vast and expanding literature on approximation
algorithms for various connectivity requirements; a landmark result is
due to Jain \cite{jain}, giving a 2-approximation algorithm for
survivable network design, a general class of edge-connectivity
requirements. For a survey on such problems, see \cite{kortsarz07,Gupta11}.

Despite the NP-completeness for general costs,  the special case of
minimum cardinality augmentation (that is, every edge on the node set
has equal cost) turned out to be polynomially tractable in several
cases and gives rise to a surprisingly rich theory. For minimum
cardinality $k$-edge-connectivity augmentation, an exact solution can
be found in polynomial time (Watanabe and Nakamura \cite{watanabe}, Frank \cite{frank-edge}). We refer the reader to the recent book by Frank \cite[Chapter 11]{frankbook} on results and techniques of minimum cardinality connectivity augmentation problems.

In connectivity orientation problems, the input is an undirected graph, and one is interested in the existence of an orientation satisfying certain connectivity requirements. For $k$-edge-connected orientations, the classical result of  Nash-Williams \cite{nash-williams-orient} gives a necessary and sufficient condition and a polynomial-time algorithm for finding such an orientation, whereas for rooted $k$-edge-connectivity, the corresponding result is due to Tutte \cite{tutte}.
A natural common generalization of these two connectivity notions is $(k,\ell)$-edge-connectivity: for integers $k\ge \ell$, a directed graph $D=(V,A)$ is said to be $(k,\ell)$-edge-connected from a root node $r_0\in V$ if for every $v\in V-r_0$ there exists $k$-edge-disjoint directed paths from $r_0$ to $v$, and $\ell$ edge-disjoint directed paths from $v$ to $r_0$. The case $\ell=k$ is equivalent to $k$-edge-connectivity whereas $\ell=0$ to rooted $k$-connectivity.
A good characterization of $(k,\ell)$-edge-connected orientability was given by Frank \cite{frank80}, see Theorem~\ref{thm:orient}. Submodular flows can be used to find such orientations (\cite{frank96}, see also \cite[Chapters 9,16]{frankbook}). The submodular flow technique also enables to solve minimum cost versions of the problem, when the two possible orientations of an edge can have different costs.

\medskip

Hence in a combined connectivity augmentation and orientation question
one wishes to find a minimum cost subgraph of a given graph that admits an orientation with a prescribed connectivity property. The simplest question is $k$-edge-connected orientability; however, this can be reduced to a pure augmentation problem. According to Nash-Williams's theorem \cite{nash-williams-orient}, a graph has a $k$-edge-connected orientation if and only if it is $2k$-edge-connected. The problem of finding the minimum cost $2k$-edge connected subgraph can be approximated within a factor of $2$~\cite{KhullerV94}. Another interesting case is to require rooted $k$-edge-connected orientability. Khanna, Naor, and Shepherd~\cite{Khanna05} give a polynomial-time algorithm for this problem and more generally, show that the problem is polynomial-time solvable if the connectivity requirements are given by a \emph{positively intersecting supermodular} function.

In this paper we consider the more general requirement of $(k,\ell)$-edge-connectivity, formally, defined as follows. Let $V\choose 2$ denote the edge set of the complete undirected graph on node set $V$.
We are given
an undirected graph $G=(V,E)$ and a cost function $c:{V\choose 2}\setminus E\rightarrow \RR_+$. The goal is to find a minimum cost subgraph $F\subseteq {V\choose 2}\setminus E$ such that $(V,E\cup F)$ admits a $(k,\ell)$-edge-connected orientation. For notational convenience, we denote by $E^*$ the set of edges that can be added to $G$ (instead of ${V\choose 2}\setminus E$). A formal description of the problem appears in the figure below. We present a $\ratio$-approximation algorithm for this problem.

\begin{center}
\vspace{.3cm}
\fbox{
\parbox{0.85\linewidth}{
\smallskip
\noindent
 Minimum Cost $(k,\ell)$-Edge-Connectivity Orientation Problem
\vspace{1mm}

\begin{tabularx}{\linewidth}{lX}
\textit{Input:}&
Undirected graph $G=(V,E)$, an edge set $E^*\subseteq {V\choose 2}$
with a
cost function
$c:E^*\rightarrow \RR_+$, 
nonnegative integers $k\geq \ell$, and root $r_0$.\\
\textit{Find:}&
 Minimum cost set of edges $F\subseteq E^*$ such that
$(V,E\cup F)$ has a $(k,\ell)$-edge connected orientation.
\end{tabularx}
}}
\vspace{.3cm}
\end{center}
\medskip

For this problem, Frank and Kir\'aly \cite{frankkiraly} gave a polynomial-time algorithm for finding the exact solution in the \emph{minimum cardinality} setting; their result employs the toolbox of splitting off techniques and supermodular polyhedral methods, used for minimum cardinality augmentation problems. They also address related questions of degree-specified augmentations and orientations of mixed graphs.

As opposed to the polynomially solvable setting of Frank and Kir\'aly \cite{frankkiraly}, our problem is NP-complete, and therefore fundamentally different techniques are needed. Our algorithm is based on iterative rounding, a powerful technique introduced by Jain \cite{jain} for survivable network design; see the recent book~\cite{LRSbook} for other results using the technique. The standard way to apply the technique is to {\em(a)} formulate a linear programming relaxation for the problem, {\em(b)} use the \emph{uncrossing technique} to show that any basic feasible solution can be characterized by a ``simple'' family of tight constraints, {\em(c)} use a counting argument to show that there is always a variable with large fractional value in any basic feasible solution. The algorithm selects a variable with large fractional value in a basic optimal LP solution and includes it in the graph. The same argument is then applied iteratively to the residual problem.

While we also use the framework of iterative rounding, our application requires a number of new techniques and ideas. Firstly, the standard cut relaxation for the problem, \eqref{lp:2}, is not amenable for the iterative rounding framework. We exhibit basic feasible solutions with no variable having a large fractional value. Instead, we use the characterization given by Frank~\cite{frank80} of undirected graphs that admit a $(k,\ell)$-edge connected orientation. This yields a different linear programming formulation~\eqref{lp}, which has constraints for partitions as opposed to cut constraints. However, the standard uncrossing technique for cuts is no longer applicable. One of the main contributions of our result is to extend the uncrossing technique to partition constraints and show that any basic feasible solution is characterized by partition constraints forming an appropriately defined tree structure (Theorem~\ref{thm:charac}).
Provided the appropriate set of constraints, the existence of an edge with high fractional value (Theorem~\ref{thm:frac-value}) is proved via the token argument originating from Jain \cite{jain}. Again, dealing with partitions requires a substantially more intricate argument; already identifying a tree structure on the tight partitions is nontrivial.

Thus we show that every basic feasible solution $x^*$ must contain an edge $e\in E^*$ such that $x^*_e\ge 1/\ratio$. We add all such edges to $F$, and iterate until we obtain a graph admitting a $(k,\ell)$-edge-connected orientation.  Our results are valid for the more general, abstract problem setting of covering nonnegative crossing $G$-supermodular functions, introduced by Frank \cite{frank80}. This extension involves extending our uncrossing techniques to collection of partitions and \emph{co-partitions} since the linear programming formulation has constraints for both partitions and co-partitions. This introduces slightly more technical challenges but gives an unified and general framework to present our results.

So far we considered symmetric orientation costs and our aim was to identify a minimum cost augmentation having a certain orientation. In a more general setting, one might differentiate between the cost of the two possible orientations in this setting. For the original orientation problem without augmentation, finding a minimum cost $(k,\ell)$-edge-connected orientation reduces to submodular flows  \cite{frank96}. The problem becomes substantially more difficult when combined with augmentation, even when the starting graph is empty, and the connectivity requirement is $k$-edge-connectivity. This problem was studied by Khanna, Naor, and Shepherd \cite{Khanna05}, posed in the following equivalent way: find a minimum cost $k$-edge-connected subgraph of a directed graph with the additional restriction that for any pair of nodes $u$ and $v$, at most one of the arcs $(u,v)$ and $(v,u)$ can be used. They gave a 4-approximation for $k=1$; there is no constant factor approximation known for any larger value of $k$. In Section~\ref{sec:mixed}, we study this problem, and exhibit an example showing that the integrality gap of the natural LP relaxation is at least $\Omega(|V|)$ even for $k=2$.

Combined orientation and augmentation settings were also studied by Cygan, Kortsarz and Nutov \cite{cygan12}, who gave a $4$-approximation algorithm for finding a minimum cost subgraph that admits a Steiner Forest Orientation, that is, an orientation containing a directed path between
a collection of ordered node pairs.

\subsection{Formal Statement of Results}

Our first main result is the following theorem which gives an approximation algorithm for the Minimum Cost $(k,\ell)$-Edge Connectivity Orientation Problem.

\begin{theorem}\label{thm:main1}
There exists a polynomial-time \ratio{}-approximation algorithm for the Minimum Cost $(k,\ell)$-Edge Connectivity Orientation Problem.
\end{theorem}

This result is obtained as a special case of a more general theorem. The general result is on covering crossing $G$-supermodular functions, introduced by Frank
\cite{frank80,frank96}. We now introduce this general framework,  starting with  notation and a few technical definitions.

For a directed graph $D=(V,A)$ and a subset $S\subseteq V$ of nodes, let $\bar S\defeq V\setminus S$ denote the complement of $S$.
For any subsets $B\subseteq A$ and $S\subseteq V$, we let
$\delta^{out}_B(S)\defeq\{(u,v)\in B: u\in S, v\notin S\}$ denote the set
of edges in $B$ which have their tail in $S$ and head outside of
$S$. We define $\delta^{in}_B(S)\defeq \delta^{out}_B(\bar S)$. We let $d_A^{out}(S)\defeq|\delta_A^{out}(S)|$ and $d_A^{in}(S)\defeq|\delta_A^{in}(S)|$.
 For an undirected  graph $G=(V,E)$ and subsets $F\subseteq E$,
$S\subseteq V$, we denote by
$\delta_F(S)$ the set of edges in $F$ with exactly one endpoint in
$S$.
We let $d_F(S)\defeq |\delta_F(S)|$, and for two
 subsets $S, T\subseteq V$, we let $d_F(S,T)$ to denote the set of
 edges in $F$ with one endpoint in  $S\setminus T$ and other in $T\setminus S$.
 For the graph $G=(V,E)$,
we shall also use $d_G(S,T)\defeq d_E(S,T)$.

The subsets $S,T\subseteq V$ are called \emph{crossing} if all four sets $S\cap T,
S\setminus T, T\setminus S$, and $V\setminus (S\cup T)$ are non-empty.
A function $f:2^V\rightarrow \ZZ_+$ is called \emph{crossing supermodular} if for all $S,T\subset V$ which are crossing we have
$$f(S)+f(T)\leq f(S\cap T) +f (S\cup T).$$
Assume we are also given an undirected graph $G=(V,E)$ on the
ground set $V$.
A function $f:2^V\rightarrow \ZZ_+$ is called \emph{crossing $G$-supermodular} if for all $S,T\subset V$ which are crossing we have
$$f(S)+f(T)\leq f(S\cap T) +f (S\cup T)+d_G(S,T).$$
Note that if a function is crossing supermodular, then it is also
crossing $G$-supermodular for any graph $G$ on the same ground set $V$.

A directed graph $D=(V,A)$ is said to \emph{cover} the function $f:2^{V}\rightarrow \ZZ_+$,  if for all subsets $S\subseteq V$,
 $$d^{in}_A(S)\geq f(S).$$
 We say that the undirected graph $(V,H)$ is \emph{$f$-orientable}, if there exists an orientation $A$ of the edges in $H$ such that $(V,A)$
 covers $f$. 

Let us now formulate the minimum-cost $(k,\ell)$-edge-connectivity
orientation problem in this framework. Given integers $k\ge \ell\ge 0$
and root  node $r_0\in V$, we let
 \begin{equation}
 f(S)\defeq\begin{cases}
k,\quad \mbox{if }r_0\notin S,\ S\neq\emptyset;\\
\ell,\quad \mbox{if }r_0\in S,\ S\neq V;\\
0,\quad\mbox{if }S=\emptyset \mbox{ or }S=V.
 \end{cases}\label{def:kl}
 \end{equation}
By Menger's theorem, a digraph covers $f$ if and only if it is
$(k,\ell)$-edge-connected from the root $r_0$. Hence, an edge set
admits a $(k,\ell)$-edge-connected orientation if and only if it is
$f$-orientable.

This function $f$ is clearly nonnegative and crossing supermodular. Indeed, let $S$ and $T$ be crossing sets. If $r_0\notin S$ and $r_0\notin T$ then $r_0$ is not in $S\cap T$ and $S\cup T$ and function takes the value of $k$ for all these sets. Similarly, if $r_0$ is in both $S$ and $T$, then $r_0\in S\cap T$ and $r_0\in S\cup T$ and thus it takes the value of $l$ on all these sets. In the remaining case, if $r_0$ is in exactly one of $S$ and $T$, then $r_0\in S\cup T$ and $r_0\notin S\cap T$ and thus the equality still holds.

\medskip

We formulate the central problem and main result of our paper.

\begin{center}
\fbox{
\parbox{0.85\linewidth}{
\smallskip

\noindent
 Minimum Cost $f$-Orientable Subgraph Problem
\vspace{1mm}

\begin{tabularx}{\linewidth}{lX}
\textit{Input:}&
Undirected graph $G=(V,E)$, an edge set $E^*\subseteq {V\choose 2}$
with a
cost function
$c:E^*\rightarrow \RR_+$,  a nonnegative valued crossing
$G$-supermodular function $f:2^{V}\rightarrow \ZZ_+$. \\
\textit{Find:}&
 Minimum cost set of edges $F\subseteq E^*$ such that
$(V,E\cup F)$ is $f$-orientable.
\end{tabularx}
}}
\end{center}

\medskip

For the function $f$, we assume that an oracle is given that returns
$f(S)$ for any $S\subseteq V$ in polynomial time. This assumption
will hold for the rest of the paper.

\begin{theorem}\label{thm:main}
There exists a polynomial-time \ratio{}-approximation algorithm for the Minimum Cost $f$-Orientable Subgraph Problem.
\end{theorem}

Theorem~\ref{thm:main} is proved in Section~\ref{sec:f-orient}. Theorem~\ref{thm:main1} follows as a corollary.

\medskip

We next consider the {\em Asymmetric Augmentation with Orientation
  Constraints Problem}, introduced by Khanna, Naor, and
Shepherd~\cite{Khanna05}.
Here we are given an undirected graph $G=(V,E)$, an integer $k$, and
costs $c_{uv}$ and $c_{vu}$ for the two possible orientations of an edge $(u,v)\in E$. The goal is
find subgraph $F\subseteq E$ such that there exists an orientation $A$
of $F$ which is $k$-edge connected. Observe that we are allowed to pick any edge at most once in $F$ and thus we can use at most one of the orientations. The objective to minimize is
$\sum_{(u,v)\in A} c_{uv}$. While as in the previous problem, an edge is allowed to be oriented in one direction, it differs in the
fact that cost of an edge depends on its orientation.
For $k=1$, Khanna, Naor, and Shepherd~\cite{Khanna05} show that the
integrality gap is upper bounded by $4$. They also conjectured that
the integrality gap is constant for all
$k$. We refute this conjecture by showing that the integrality gap can
be $\Omega(|V|)$ already for $k=2$.
The integrality gap will be given for a special case of this problem,
called  \emph{Augmenting a mixed graph with orientation constraints}, described below.

\begin{center}
\vspace{.5cm}
\fbox{
\parbox{0.85\linewidth}{
\smallskip

\noindent
Augmenting a Mixed Graph with Orientation Constraints
\vspace{1mm}

\begin{tabularx}{\linewidth}{lX}
\textit{Input:}&
Mixed graph $G=(V,A\cup E)$, with $A$ being a set of directed and
$E$ a set of undirected edges,  a further set of undirected edges
                 $E^*$
with a cost function $c:E^*\rightarrow \RR_+$, and an integer $k$.\\
\textit{Find:}&
 Minimum cost set of edges $F\subseteq E^*$ such that
 $E\cup F$ admits an orientation $H$ for which
$(V,A\cup H)$ is $k$-edge-connected.
\end{tabularx}
}}
\vspace{.5cm}
\end{center}

This problem corresponds to the special case of the
{\em Asymmetric Augmentation with Orientation
  Constraints Problem}
when there are two types of edges $(u,v)$. Either the cost is
symmetric, i.e., $c_{uv}=c_{vu}$ (i.e. the edges in $E^*\cup E$) or,
 one of $c_{uv}$ is 0
and the other is $\infty$ (i.e. the edges in $A$).
{\em Augmenting a Mixed Graph with Orientation Constraints} seems to be a mild extension only of  the $f$-orientable
subgraph problem; moreover, we have the simpler requirement of
$k$-edge-connectivity. However, the  mixed graph setting
leads to a substantially more difficult setting already in this
simplest case.

\begin{theorem}\label{thm:orientation}
For any $k\geq 2$, there is an instance of the
{\em Augmenting a mixed graph with orientation constraints} such that the integrality gap of the natural linear programming formulation is $\Omega(|V|)$.
\end{theorem}

The natural LP relaxation refers to \eqref{lp:3} in Section~\ref{sec:mixed};
Theorem~\ref{thm:orientation} is proved in the same section.

\section{Approximation Algorithm for $f$-orientable subgraph}\label{sec:f-orient}

In this section we prove Theorem~\ref{thm:main}. We start by giving a natural linear programming formulation~\eqref{lp:2} for the problem and note that it is not amenable to iterative rounding. Using a result of Frank~\cite{frank80} we give a different linear programming relaxation~\eqref{lp} that will be used for the algorithm. The feasible region of the new linear program is the projection of \eqref{lp:2} to an appropriate set of variables. The iterative rounding algorithm is given in Section~\ref{sec:alg}.
Theorem~\ref{thm:frac-value} asserts that every extreme point solution to \eqref{lp} has a component of fractional value at least $\frac 16$. Theorem~\ref{thm:main} then follows in the standard way.

To derive Theorem~\ref{thm:frac-value}, we first give a characterization of any extreme point via a simple family of tight constraints in Theorem~\ref{thm:charac}. The main ingredient in this characterization is a new uncrossing argument for partition and co-partition constraints, described in Section~\ref{sec:uncrossing}.   The proof of Theorem~\ref{thm:frac-value} is completed in Section~\ref{sec:iterative} via a counting argument building on this characterization.

\subsection{Linear Programming Formulations}

A natural linear programming relaxation for the minimum cost $f$-orientable subgraph problem is \eqref{lp:2} where we have $x$ variables for edges in $E^*$ and $y$ variables for the two possible orientations for each edge in $E\cup E^*$. For an integer solution, the $x$ variables are set to one for edges in $F$ and the $y$ variables are set to one for the appropriate orientation in $A$ of edges $E\cup F$ covering $f$.

\begin{equation}\tag{{\em
        LP1}} \label{lp:2}
\begin{aligned}
    \mbox{ minimize} \sum_{(u,v)\in E^*} c_{uv} x_{uv} & &\notag \\
    \mbox{ s.t.}~~
     y(\delta^{in}(U))& \geq f(U)  &  \;\; \forall\; U\subset V\notag\\
     y_{uv}+y_{vu} &= x_{uv}& \;\; \forall\; (u,v)\in E^*\\
      y_{uv}+y_{vu}&= 1 &\;\; \forall (u,v)\in E\notag\\
 \textbf{0}\le x&\le \textbf{1} &\notag\\
y&\ge \textbf{0}
\end{aligned}
\end{equation}

Surprisingly, we find that the~\eqref{lp:2} is not amenable to
iterative rounding and there are basic feasible solutions where each
non-integral $x$ variable is ${O(\frac{1}{n})}$; see
Remark~\ref{rem:lp:2} for such an example. We observe that such
basic feasible solutions are due to the presence of variables $y$,
which are auxiliary to the problem since they do not appear in the
objective. Thus we formulate an equivalent linear program where we
project the feasible space of \eqref{lp:2} onto the $x$-space. Before describing the linear program, we introduce some notation.

For a family
$\F$ of subsets of $V$,  we let $\bar{\F}\defeq\{\bar S: S\in\F\}$ denote
the set of complements. For two sets $\F$ and $\H$, let  $\F\dunion \H$ denote the multi set arising as
the disjoint union -  that is, elements occurring in both sets are
included with multiplicity. A collection
$\P$ of subsets of $V$ is called a \emph{partition} if every element
of $V$ is in exactly one set in $\P$ and it is called a \emph{co-partition}
if every element of $V$ is in all but one set in $\P$.
The subsets comprising a partition or co-partition will be referred to
as its \emph{parts}.
 If $\P$ is a co-partition, then $\bar \P=\{\bar S: S\in \P\}$ forms a partition, called the \emph{complement partition} of $\P$. Similarly, if $\P$ is a partition then $\bar \P$ forms a co-partition.
For any collection $\P$ of subsets of $V$, we denote $\chi_F(\P)$ to be
the set of edges $(u,v)\in F$ for which there exists two distinct sets
$S,T \in \P$ such that $u\in S\setminus T$ and $v\in T\setminus S$,
and we let $e_F(\P)\defeq
|\chi_F(\P)|$.
 For the graph $G=(V,E)$,
we shall also use $e_G(\P)\defeq e_E(\P)$.

The new linear program is based on the characterization of $f$-orientable graphs
given by Frank~\cite{frank80}.

\begin{theorem}[Frank, \cite{frank80}]\label{thm:orient}
Let $G=(V,E)$ be an undirected graph, and let $f:2^V\rightarrow
\mathbb{Z}$ be a nonnegative valued crossing $G$-supermodular function
with $f(\emptyset)=f(V)=0$.
Then $G$ is $f$-orientable if and only if for every partition and every co-partition
$\P$ of $V$,
\[e_{G}(\P)\geq \sum_{S\in \P}f(S).\]
\end{theorem}

Note that necessity follows easily: for a partition or co-partition
$\P$, the edges in $\chi_G(\P)$ need to cover the sets in
$\P$ with a total demand  $ \sum_{S\in \P}f(S)$, and each  edge can
contribute to covering one of the parts in both possible orientations.

Theorem~\ref{thm:orient} verifies that the following is a valid linear programming relaxation of the Minimum Cost $f$-Orientable Subgraph Problem: requiring all $x_e$ values integer provides an augmenting edge set $F$.
\begin{equation}\tag{{\em LP2}} \label{lp}
 \begin{aligned}
   \mbox {minimize} \sum_{e \in E^*} c_e x_e & & \\
    \mbox{s.t.}~~
     x(\chi_{E^*}(\P))& \geq \sum_{S\in \P}f(S)-e_{G}(\P) &  \;\;\\
     & \forall\; \textrm{ partition or co-partition } \P \textrm{ of } V\notag;\\
   \textbf{0}\le x&\le \textbf{1} &\notag
\end{aligned}
\end{equation}
We use the Ellipsoid method \cite{groetschel81} to solve \eqref{lp},
by providing a separation oracle.
Theorem~\ref{thm:orient} implies that the feasible region of
\eqref{lp} is the projection of the feasible region of \eqref{lp:2} to
the $x$-space\footnote{This follows due to a standard scaling argument. Let $x$ denote a feasible solution to \eqref{lp} and let $M$ be a large integer such that $M\cdot x$ is an integral vector. Applying Theorem~\ref{thm:orient} gives an orientation of the obtained multi-graph satisfying requirement function $M\cdot f$. Taking each oriented edge to a fraction of $\frac1M$, we obtain a fractional solution to \eqref{lp:2}. The other direction showing that for any $(x,y)$ that is feasible to \eqref{lp:2}, $x$ is feasible for \eqref{lp} follows from a simple counting.}.
While \eqref{lp} still has an exponential number of constraints, it
can be further reduced to a submodular flow problem \cite{frank96} which can be solved efficiently using the value oracle for function $f$. 

\subsection*{$(k,\ell)$-edge connectivity} Consider the special case when the crossing $G$-super\-modular function corresponds to the requirement for $(k,\ell)$-edge connectivity, with
$f$ defined as in \eqref{def:kl}. Using $k\ge \ell$, the co-partition
inequalities become redundant in~\eqref{lp}, and we obtain the following theorem.

\begin{theorem}[Frank, \cite{frank80}]
An undirected graph $G=(V,E)$ admits a $(k,\ell)$-edge-connected
orientation for $k\ge \ell$ with root $r_0\in V$ if and only if
\[e_{G}(\P)\geq k(|\P|-1)+\ell,\]
for every partition $\P$ of $V$.
\end{theorem}

A remarkable consequence is that $(k,\ell)$-edge-connected
orientability is a property independent of the choice of the root node $r_0$.
 Specializing even further, for $k=\ell$ it is equivalent to
 $k$-edge-connectivity, and it is easy to see that the condition
can be further simplified to $d_G(S)\ge 2k$ for every $S\subseteq V$, that
 is, the graph is $2k$-edge-connected. This gives the classical
 theorem of Nash-Williams \cite{nash-williams-orient}. For $\ell=0$, the
 above theorem is equivalent to Tutte's theorem \cite{tutte} on rooted
 $k$-edge-connected orientability.

\subsubsection{Iterative Rounding Algorithm}\label{sec:alg} We use the iterative rounding
algorithm in Figure~\ref{alg:ir}, as introduced by Jain~\cite{jain}.
Further, note that if a function is crossing $G$-supermodular for a
graph $G=(V,E)$, than it is crossing $G'$-supermodular for every
$G'=(V,E')$ if $E'\supseteq E$.
Theorem~\ref{thm:frac-value}
guarantees that $E'$ is strictly extended in every iteration, hence
the algorithm terminates in $O(|E^*|)$ iterations. The argument showing
that this is a $\ratio$-approximation follows the same lines as in \cite{jain} assuming Theorem~\ref{thm:frac-value}.

\begin{figure*}[ht]
\centering
\fbox{\parbox{\textwidth}{
\textsl{Input:} 
Undirected graph $G=(V,E)$, an edge set $E^*\subseteq {V\choose 2}$
with a
cost function
$c:E^*\rightarrow \RR_+$,  a nonnegative valued crossing
$G$-supermodular function $f:2^{V}\rightarrow \ZZ_+$.\\
\textsl{Output:} An $f$-orientable graph $(V,E')$ with $E\subseteq
E'\subseteq E\cup E^*$.
\begin{enumerate}
\item $E'\gets E$.
\item While $E^*\neq \emptyset$
\begin{enumerate}
\item Solve \eqref{lp} to obtain a basic optimal solution $x^*$.
\item $E^*\gets E^*\setminus \{e: x^*_e=0\}$.
\item $E'\gets E'\cup \{e\in E^*: x^*_e\geq\frac1\ratio\}$.
\end{enumerate}
\item Return $(V,E')$.
\end{enumerate}
}}
\caption{Iterative rounding algorithm}\label{alg:ir}
\end{figure*}

\begin{theorem}\label{thm:frac-value}
Let $x^*$ be an extreme point solution to \eqref{lp} where $f:2^V\rightarrow
\mathbb{Z}$ is a nonnegative valued crossing $G$-supermodular function
with $f(\emptyset)=f(V)=0$. Then there exists an edge $e$ such that $x^*_e\geq \frac{1}{\ratio}$.
\end{theorem}

The rest of the section is devoted to the proof of Theorem~\ref{thm:frac-value}. Section~\ref{sec:cross-free} introduces the appropriate notion
of cross-freeness for partitions and co-partitions that will
characterize extreme point solutions of (\ref{lp}). This characterization is given in Section~\ref{sec:extreme} and will rely on a new uncrossing that we introduce for partitions and co-partitions that
 will be given in Section~\ref{sec:uncrossing}. In Section~\ref{sec:partial}, we show a partial order that can be derived from this cross-free family of partition and co-partitions. Results of these sections will be needed to
prove Theorem~\ref{thm:frac-value} in Section~\ref{sec:iterative}.

\subsection{Strongly Cross-Free Family}\label{sec:cross-free}
A family $\F$ is \emph{laminar} if for any two sets $S,T\in
\F$, either they are disjoint or one contains the other.
A collection of sets $\F$ is
\emph{cross-free} if no two sets $S,T\in \F$ are
crossing.
Note that a cross-free family is non-laminar if there exists $S,T\in \F$ with $S\cap T\neq\emptyset$, $S\cup T=V$.
Observe that if a family $\F$ is cross-free, it remains so after adding the complements of some of the sets.

For two partitions or co-partitions $\P$ and $\Q$, we say that $\P$ and $\Q$ are {\em cross-free}, if $\P\cup \Q$ is a cross-free family. This is a natural and desirable property of the family of tight constraints; however, we will need a stronger notion, called {\em strongly cross-free}, as introduced later in this section.

\begin{remark}\em
The strongly cross-free family differs from other structured families obtained when uncrossing partitions, for example, in the work of Chakrabarty, K\"onemann and Pritchard~\cite{ChakrabartyKP13} (see also Chapter 48.2 and Chapter 49.6 in Schrijver~\cite{schrijver-book} for other examples).
In these works, partitions can be uncrossed to their meet and join in the partition lattice. (The partition lattice is defined by the partial order $\R<\R'$ if $\R$ is a refinement of $\R'$.) That is, any two partitions $\P$ and $\Q$ can be uncrossed to two partitions $\P\wedge \Q$ and $\P\vee \Q$ such that the components of  $\P\wedge \Q$  are the intersections of the components of $\P$ and $\Q$, and the components of $\P\vee\Q$ are the connected components of the hypergraph $\P\cup\Q$. Such an uncrossing is not possible in our framework.
Figure~\ref{fig:uncrossing3} illustrates an example where two partitions in our setting cannot be uncrossed while such partitions would have been uncrossed in \cite{ChakrabartyKP13,schrijver-book}.  The reason is that our requirement function $f$ is crossing $G$-supermodular as compared to the fully supermodular case in ~\cite{ChakrabartyKP13,schrijver-book}. For partitions $\P$ and $\Q$, if there are parts $P\in \P$ and $Q\in \Q$ such that $P\cup Q=V$ then the constraints for $\P$ and $\Q$ cannot be uncrossed to the meet and join in the partition lattice.

In our notion of strongly cross-free partitions, we require that $\P\cup \Q$ is a cross-free family, but we cannot require the stronger property that $\P\cup\Q$ is laminar. This is in contrast with the above cited works, where $\P$ and $\Q$ are uncrossed simply if one is a refinement of the other, hence
$\P\cup\Q$ is a laminar family.
\end{remark}

\begin{figure*}[tp]
\centering    \includegraphics[width=5cm, height=3.5cm]{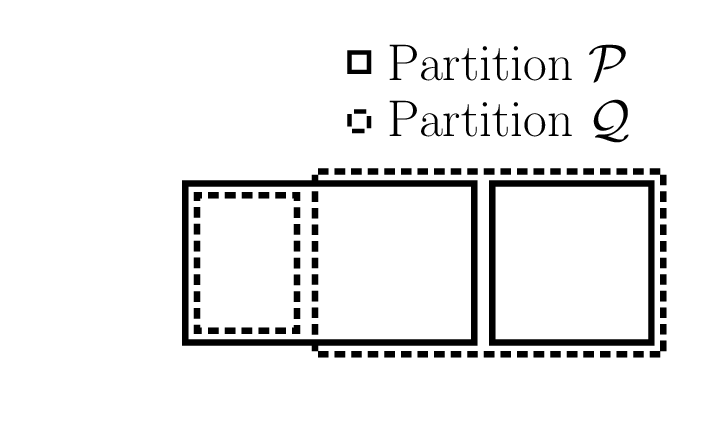}
\caption{Partitions $\P$ and $\Q$ whose parts intersect but do not cross. $\P$ and $\Q$ cannot be uncrossed in our setting. } \label{fig:uncrossing3}
\end{figure*}

\medskip

Some further notation and definitions are in order before we introduce the notion of strongly cross-freeness.
By a {\em sub-partition} of $V$ we mean a collection of disjoint subsets of $V$ (equivalently, a partition of a subset of $V$).
The sets in a sub-partition will also be referred to as {\em parts}.
 For a sub-partition $\P$, let  $\supp(\P)$ denote the union of parts of $\P$.
Two sub-partitions $\P$ and $\Q$ are \emph{disjoint} if $\supp(\P)\cap\supp(\Q)=\emptyset$.

Let us fix a special node $\r\in V$, called the {\em root node}. This can be chosen arbitrarily, but remains fixed throughout the  argument.
 We use the following notational convention. For every partition $\P=\{P_1,\ldots, P_p\}$ of $V$, $P_1$ (that is, the part indexed 1)  contains $\r$; whereas for every co-partition  $\P=\{P_1,\ldots, P_p\}$, $P_1$ is the single part not containing \r.
With every partition or co-partition $\P$ of $V$, we associate a sub-partition $\tilde\P$ as follows.
\[
\tilde\P\defeq\begin{cases}
\{P_2,\ldots,P_p\},&\mbox{if $\P$ is a partition};\\
\{\bar P_2,\ldots,\bar P_p\},&\mbox{if $\P$ is a co-partition}.
\end{cases}
\]
Note that if $\P$ is a partition then $\supp(\tilde \P)=\bar P_1$, whereas if $\P$ is a co-partition then $\supp(\tilde \P)=P_1$.
The next claim summarizes simple properties of cross-free partitions and co-partitions.
\begin{claim}\label{cl:sup-lam}
Assume $\P$ and $\Q$ are cross-free partitions or co-partitions. For the sub-partitions $\tilde\P$ and $\tilde\Q$ as defined above, $\tilde\P\cup \tilde\Q$ is a laminar family. Further, the sets
$\supp(\tilde\P)$ and $\supp(\tilde\Q)$ are either disjoint, or one of them contains the other.
\end{claim}
\begin{proof}
As noted above, complementing sets in a cross-free family keeps cross-freeness, hence $\tilde\P\cup \tilde\Q$ is cross-free. The root $r$ is not contained in any of these sets, implying laminarity. This also implies the last claim.
\end{proof}

\medskip

The sub-partition
$\Q$  \emph{dominates} the sub-partition $\P$, if every part of $\P$ is
a subset of some part of $\Q$. We further say that $\Q$ {\em strongly dominates} $\P$, if every part of $\P$ is a subset of the same part $Q\in\Q$, that is, $\supp(\P)\subseteq Q$ for some $Q\in\Q$. The following transitivity properties are straightforward and allows us to deal with partitions and co-partitions in an unified manner in many cases.
\begin{claim}\label{cl:trans}
Consider sub-partitions $\P$, $\Q$ and $\R$. If $\Q$ dominates $\P$ and $\R$ dominates $\Q$, then $\R$ also dominates $\P$. Further, if $\Q$ dominates $\P$ and $\R$ strongly dominates $\Q$, then $\R$ strongly dominates $\P$. Similarly, if
 $\Q$ strongly dominates $\P$ and $\R$ dominates $\Q$, then $\R$ strongly dominates $\P$.
\end{claim}

\medskip

We are ready to introduce the notion of strong cross-freeness.
\begin{definition}
Consider two partitions or co-partitions $\P$ and $\Q$ with associated sub-partitions $\tilde\P$ and $\tilde\Q$. We say that $\P$ and $\Q$ are
 \emph{strongly cross-free}, if
\begin{itemize}
\item $\P\cup \Q$ is a cross-free family, and further
\item if $\P$ and $\Q$ are both partitions or both co-partitions, then the associated sub-partitions $\tilde\P$ and $\tilde\Q$ are either disjoint or one of them dominates the other;
\item  if one of $\P$ and $\Q$ is a partition and the other is a co-partition, then
$\tilde\P$ and $\tilde\Q$ are either
 disjoint or one of them strongly dominates the other.
\end{itemize}
A family $\aP$ of partitions and co-partitions is strongly cross-free,
if any two of its members are strongly cross-free.
\end{definition}

\begin{figure*}[tp]
  \begin{minipage}[b]{0.45\linewidth}
    \includegraphics[width=6cm, height=6cm]{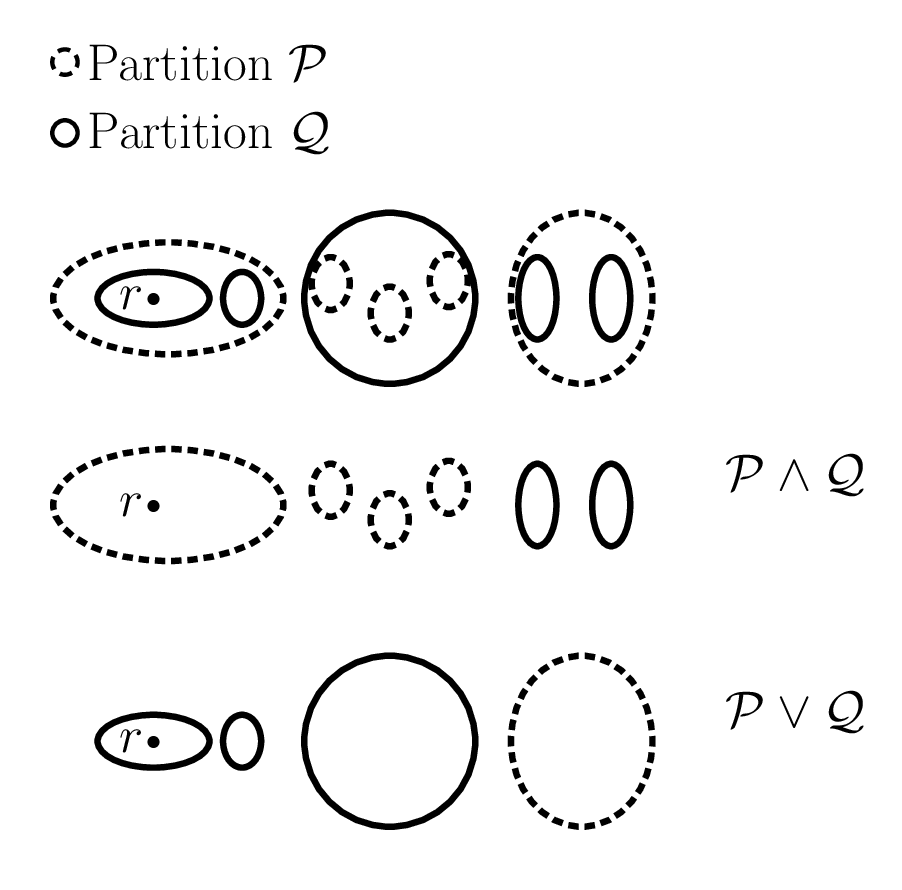}
\caption{Uncrossing two partitions $\P$ and $\Q$.} \label{fig:uncrossing1}
  \end{minipage}
  \hspace{0.5cm}
  \begin{minipage}[b]{0.55\linewidth}
    \includegraphics[width=7cm, height=7cm]{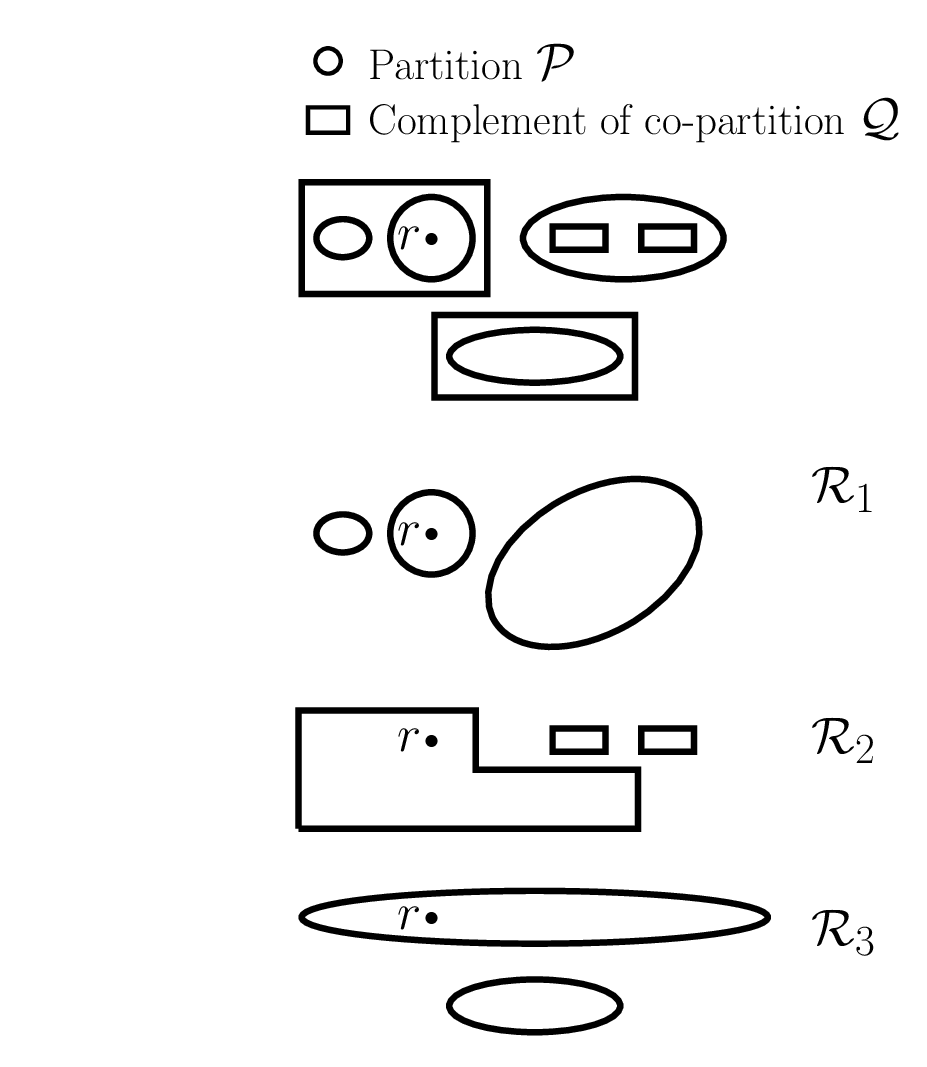}
\caption{Uncrossing partition $\P$ and co-partition $\Q$.} \label{fig:uncrossing2}
 \end{minipage}
\end{figure*}

As an example, consider Figures~\ref{fig:uncrossing1} and
\ref{fig:uncrossing2}. In Figure~\ref{fig:uncrossing1}, the two
partitions $\P$ and $\Q$ are {\em not} strongly cross-free since none dominates the other and $\tilde\P$ and $\tilde \Q$ are not disjoint. This is despite the fact that
$\P\cup\Q$ is a cross-free family. However, both of them are strongly
cross-free with both $\P\vee \Q$ and $\P\wedge\Q$, which are also
strongly cross-free with each other (these uncrossed sub-partitions will be
defined in Section~\ref{sec:cross-free}).  Figure~\ref{fig:uncrossing2}
illustrates a partition $\P$ and a co-partition $\Q$ that are not
strongly cross-free, yet both of them are strongly cross-free with all
of $\R_1$, $\R_2$ and $\R_3$. Also, the $\R_i$'s are pairwise strongly cross-free.

It may lead to ambiguity that a partition comprising two sets is also a co-partition. However, note that cross-freeness and strongly cross-freeness coincide if either $\P$ or $\Q$ comprises two sets, and therefore this does not make a difference. More generally, to differentiate cross-freeness from strongly cross-freeness, we will call two partitions or co-partitions $\P$ and $\Q$ \emph{weakly cross-free} whenever they are cross-free (that is $\P\cup \Q$ is a cross-free family of sets), but not strongly cross-free.

The following claim provides an alternative view of a strongly cross-free  partition and co-partition.
\begin{claim}\label{cl:strongly-cf}
 If $\P$ is a partition and $\Q$ is a co-partition, then $\P$ and $\Q$ are strongly cross-free if and only if  there exist parts $P\in \P$ and $Q\in \Q$ such that $Q\subseteq P$.
\end{claim}
\begin{proof}
Assume first $\P$ and $\Q$ are strongly cross-free. Recall that $\supp(\tilde\P)=\bar P_1$ and $\supp(\tilde\Q)=Q_1$. If
$\tilde\P$ and $\tilde\Q$ are disjoint, then $\bar P_1\cap Q_1=\emptyset$ gives $Q_1\subseteq P_1$. If $\tilde\Q$ strongly dominates $\tilde\P$, that is, $\bar P_1=\supp(\tilde\P)\subseteq Q'$ for some $Q'\in\tilde\Q$. By definition, $Q'=\bar Q$ for some $Q\in\Q$, and hence $\bar P_1\subseteq \bar Q$, that is, $Q\subseteq P_1$. Finally, assume $\tilde\P$ strongly dominates $\tilde\Q$. Then $Q_1= \supp(\tilde\Q)\subseteq P$ for some $P\in \tilde\P$, which is also a part of $\P$.
The converse direction is also easy to verify; note that $Q\subseteq P$ implies that $\P\cup \Q$ is cross-free.
\end{proof}

As a consequence of the above claim, if $\P$ is a partition and $\Q$ is a co-partition, then strongly cross-freeness is independent of the choice of the root $r$.

Before we end this section, we give a couple of technical lemmas that are useful in uncrossing of partitions and co-partitions. A collection $\cal F$ of subsets of the ground set $V$ is called \emph{$t$-regular} for some positive integer $t$, if every element of $V$ is contained in exactly $t$ members of $\cal F$. The family $\cal F$ is called regular, if it is $t$-regular for some value $t$.
For example, a partition is a 1-regular family, and a co-partition of cardinality $t+1$ is a $t$-regular family. The following lemma reveals the distinguished role  partitions and co-partitions play.

\begin{lemma}[{\cite[Lemma 15.3.1]{frankbook}}]\label{lem:decompose}
Every cross-free regular family $\cal F$ of $V$ can be decomposed into the disjoint union of partitions and co-partitions of $V$.
\end{lemma}

For two families
 $\F$ and $\H$ of subsets of $V$, we let $\nu(\F,\H)$ denote the
 number of pairs $X\in \F$ and $Y\in \H$ such that $X$ and $Y$
 cross. The following claim is a standard ingredient of uncrossing
 arguments, and easy to verify.
\begin{claim}\label{cl:uncross-decrease}
Let $A,B$ and $C$ be arbitrary subsets of $V$. Then
$\nu(\{A,B\},\{C\})\ge \nu(\{A\cap B, A\cup B\},\{C\})$.
\end{claim}

\subsection{Extreme Point Solution}\label{sec:extreme}

The iterative algorithm removes every edge $e$ from $E^*$ with $x^*_e=0$; consequently, we may assume that $E^*=\supp(x^*)$.
 By a slight abuse of notation, for any set of edges $F\subseteq E^*$, we will use $F$ to denote the indicator vector in $\RR^{E^*}$ of the edge set $F$. Thus, for any partition or co-partition $\P$, $\chi_{E^*}(\P)$ will also denote the indicator vector for the edge set $\chi_{E^*}(\P)$. Similarly, for a subset $S\subset V$, $\delta(S)$ will also denote the indicator vector for the corresponding edge set.  For a family of partitions and co-partitions $\aP$, we let span$(\aP)$ denote the vector space generated by the vectors $\{\chi_{E^*}(\P):\P\in \aP\}$.
An important component in the proof of Theorem~\ref{thm:main} is the following characterization of extreme point solutions.

\begin{theorem}\label{thm:charac}
Let $x^*$ be an extreme point solution to (\ref{lp}), and assume $E^*=\supp(x^*)$ and $x^*_e<1$ for each $e\in E^*$. Then there exists a family of tight partitions and co-partitions $\aP$ such that
\begin{enumerate}
\item The family $\aP$ is strongly cross-free.
\item The vectors $\{\chi_{E^*}(\P):\P\in \aP\}$ are linearly independent.
\item  $|E^*|= |\aP|$.
\end{enumerate}
\end{theorem}

Note that requiring only properties 2 and 3 is equivalent to the definition of an extreme point solution. The challenge is to enforce the strongly cross-free structure. This will rely on two uncrossing arguments, one standard and one that will appear in Section~\ref{sec:uncrossing}. We state the main lemma that we will prove in Section~\ref{sec:uncrossing}; this will allow us to complete the proof of Theorem~\ref{thm:charac}. We start by giving a few definitions and a technical claim first.

Consider a feasible solution $x:E^*\rightarrow \RR_+$ to (\ref{lp}), and a collection $\F$ of subsets of $V$. Let
$\Delta_{x}({\F})\defeq\sum_{S\in{\F}}x(\delta(S))$ denote the sum of $x$-degrees of sets in $\F$. Similarly, let $\Delta_{E}({\F})\defeq\sum_{S\in{\F}}d_E(S)$.
Note that if $\P$ is a partition or a co-partition, then $\Delta_{x}({\P})=2x(\chi_{E^*}(\P))$ and $\Delta_{E}({\P})=2e_G(\P)$.
Using this notation, let us define
\begin{align*}
\Psi_{x}(\F)&\defeq\frac 12\Delta_{x}({\F})- \sum_{S\in\F}f(S)+\frac12 \Delta_{E}({\F})\\
&=\sum_{S\in\F}\left( \frac12x(\delta(S))-f(S)+\frac12d_E(S)\right).
\end{align*}

\begin{claim}\label{cl:tight-fam}
For every feasible solution $x$ to (\ref{lp}) and every cross-free regular family $\F$, $\Psi_{x}(\F)\ge 0$ holds.
If $\Psi_{x}(\F)=0$, then for any decomposition of $\F$ into a disjoint union of partitions and co-partitions, all these (co-)partitions must be tight.
\end{claim}
\begin{proof}
By Lemma~\ref{lem:decompose}, $\F$ can always be decomposed in to a disjoint union of partitions and co-partitions; let $\F=\cup_{i=1}^t \P^i$ denote such a decomposition.
By the feasibility of $x$ we have $x(\chi_{E^*}(\P^i))\ge f(\P^i)-e_G(\P^i)$
for every $i=1,\ldots t$. Then $\Psi_{x}(\F)\ge 0$  follows by summing up these inequalities.
Further, in the case of equality we must have had equality for all $\P^i$'s, that is, all of them are tight.
\end{proof}

We now state the main lemma that will be derived by uncrossing weakly cross-free partitions/co-partitions.
\begin{lemma}\label{lem:weak-uncross}
For any $\P$ and $\Q$ weakly cross-free partitions or co-partitions, there exists a set of partitions and co-partitions $\Upsilon(\P,\Q)$ with the following properties. let  $\R\in
\Upsilon(\P,\Q)$ be arbitrarily chosen.
\begin{enumerate}[(i)]
\item
If $\P$ and $\Q$ are both tight for some feasible solution
$x$ to (\ref{lp}), then $\R$ is also tight.
\item $\R$ is strongly cross-free with both $\P$ and $\Q$.
\item $\chi_{E^*}(\P)+\chi_{E^*}(\Q)=\sum\{\chi_{E^*}(\R): \R\in\Upsilon(\P,\Q)\}$.
\item If a partition or co-partition $\S$ is strongly cross-free with both
$\P$ and $\Q$, then it is also strongly cross-free with $\R$.
\end{enumerate}
\end{lemma}

We will give the definition of $\Upsilon(.,.)$ and proof of Lemma~\ref{lem:weak-uncross} in Section~\ref{sec:uncrossing} but the statement of the lemma is enough to give the proof of Theorem~\ref{thm:charac}.

{\em Proof of Theorem~\ref{thm:charac}.}
Let $\aS$ denote the family of all tight partitions and co-partitions, and
consider the linear space $\mbox{span}({\aS})$ generated by their characteristic vectors.
Consider a strongly cross-free family ${\aP}\subseteq{\aS}$ such that the vectors $\{\chi(\P):\P\in \aP\}$ are linearly independent; assume $\aP$ is maximal for containment.
We show that $\mbox{span}(\aP)=\mbox{span}({\aS})$, that is, $\aP$
generates the entire linear space of tight partitions and
co-partitions. Property 3 then follows since the dimension of
$\mbox{span}(\aS)$ is equal to $|E^*|$, as $x^*$ is an extremal point,
and $\supp(x^*)=E^*$ is enforced by the algorithm.
For a partition or co-partition $\Q$, let  $\mu(\Q,\aP)$ denote the number of
partitions and co-partitions in the family $\aP$ that are weakly
cross-free with $\Q$.

Recall that for two sets $\F$ and $\H$, $\F\dunion \H$ denotes the multi set arising as
the disjoint union -  that is, elements occurring in both sets are
included with multiplicity.
For a contradiction, assume there exists a tight partition or co-partition $\Q$ with
$\chi(\Q)\notin \mbox{span}(\aP)$. Pick $\Q$ such that
$\nu(\Q,\dunion\aP)$ is minimal, and subject to this, $\mu(\Q,\aP)$ is
minimal.
We show that both quantities equal 0. This gives a contradiction as it
means that $\Q$ is strongly cross-free with all members of $\aP$ and
thereby contradicts the maximal choice of $\aP$.

First, assume $\nu(\Q,\dunion\aP)>0$. In particular, there exists a
$\P\in\aP$ with $\nu(\Q,\P)>0$, that is, $\P$ and $ \Q$ have some
crossing parts. Set $\F\defeq\P\dunion \Q$.
Let us transform $\F$ to a family $\F'$ by the following
uncrossing method: whenever $\F$ contains two crossing sets $A$ and
$B$, remove them and replace them by $A\cup B$ and $A\cap B$ ($\F$ is
also a multi set, i.e. it may contain multiple copies of the same set).

\begin{claim}\label{cl:uncross-finite}
The uncrossing method is finite and delivers a regular cross-free
$\F'$ with $\Psi_{x^*} (\F')=0$, and $\sum_{S\in \F}\delta(S)=\sum_{S\in\F'}\delta(S)$.
\end{claim}
\begin{proof}
The procedure must be finite since $\sum_{A\in\F} |A|^2$ strictly
increases in every step.
Initially, $\F$ is regular, being the union of two
partitions/co-partitions, and $\Psi {x^*} (\F)=0$ since both $\P$ and $\Q$
were tight. For simplicity of notation, let $\F$ and $\F'$ denote in
the following the family before and after replacing two crossing sets
$A$ and $B$ by $A\cup B$ and $A\cap B$ (as opposed to the initial and
the final family).

We prove that whenever $\F$ satisfies the above properties, then so
does $\F'$. First, regularity is maintained since
every node $s\in V$ is covered by $\{A\cup B,A\cap B\}$ equal number of
times as by $\{A,B\}$. Let us prove $\Psi_{x^*}(\F)\ge \Psi_{x^*} (\F')$.

Observe that for any two crossing sets $S$ and $T$ we have,
\begin{align*}
\frac12\delta_E(S)+\frac12\delta_E(T)&\geq \frac12\delta_E(S\cap T)+\frac12\delta_E(S\cup T) + d_E(S,T)\\
f(S)+ f(T)&\leq f(S\cap T) +f(S\cup T) +  d_E(S,T)
\end{align*}
where the first inequality follows simply from counting edges on both sides of the inequality and the second inequality follows since $f$ is $G$-supermodular. Subtracting the first inequality from the second implies that
 the function $f(S)-\frac12\delta_E(S)$ is crossing supermodular. Moreover, the function $x^*(\delta (S))$ is
submodular and
consequently,
$\frac12\delta_{x^*} (S)-f(S)+\frac12\delta_E(S)$ is crossing
submodular. Thus we have $\Psi_{x^*} (\F)\ge \Psi_{x^*} (\F')$; then $\Psi_{x^*} (\F')=0$
follows by Claim~\ref{cl:tight-fam}. Further, the equality also yields
$\delta(A)+\delta(B)=\delta(A\cup B)+\delta(A\cap B)$, and
therefore
 $\sum_{S\in \F}\delta(S)=\sum_{S\in\F'}\delta(S)$.
\end{proof}

Let us apply Lemma~\ref{lem:decompose} for the regular cross-free
family $\F'$ resulting by the above procedure, decomposing it into a
set of partitions and co-partitions
$\R^1,\R^2,\ldots,\R^\ell$. Claim~\ref{cl:tight-fam}  together with
$\Psi_{x^*} (\F')=0$ implies that all $\R^i$'s are tight. Further, we
have $\sum_{S\in\F}\delta(S)=2(\chi(\P)+\chi(\Q))$, and
$\sum_{S\in\F'}\delta(S)=2\sum_{i=1}^\ell\chi(\R^i)$. By the
previous claim,
\[
\chi(\P)+\chi(\Q)=\sum_{i=1}^\ell\chi(\R^i).
\]
By the assumption $\chi(\Q)\notin \mbox{span}(\aP)$, there exists
$1\le r\le \ell$ with $\chi(\R^r)\notin \mbox{span}(\aP)$.
The next claim gives a contradiction to the extremal choice of $\Q$.
\begin{claim}\label{cl:nubound}
$\nu(\Q,\dunion\aP)>\nu(\R^r,\dunion\aP)$
\end{claim}

\begin{proof}
First, observe that
$\nu(\F,\dunion\aP)=\nu(\Q,\dunion\aP)$, since $\F=\P\cup\Q$ and
$\dunion\aP$ is cross-free. Using Claim~\ref{cl:uncross-decrease},
$\nu(\F,\dunion\aP)$ cannot increase during the uncrossing procedure.
But in the very first step, it must strictly decrease. Indeed, the
first step uncrosses some part $P_i\in \P$ with some $Q_{j}\in\Q$.
Now $\nu(\{P_i,Q_{j}\},\{P_{i}\})=1$ but $\nu(\{P_i\cup
Q_{j},P_i\cap Q_{j}\},\{P_{i}\})=0$, and therefore $\nu(\F,\dunion
\aP)$ must strictly decrease in the first step.
Hence for the final $\F'$,
\begin{align*}
\nu(\R^r,\dunion\aP)&\le
\nu(\dunion_{i=1}^\ell\R^i,\dunion\aP)=\nu(\F',\dunion\aP)\\
&<\nu(\F,\dunion \aP)=\nu(\Q,\dunion\aP).
\end{align*}
\end{proof}

This completes the proof of $\nu(\Q,\dunion\aP)=0$. Next, assume
$\P\cup \Q$ is cross-free for every $\P\in \aP$, nevertheless,
$\mu(\Q,\aP)>0$, that is, $\P$ and $\Q$ are weakly cross-free for some
$\P\in\aP$.
Consider now the set $\Upsilon(\P,\Q)=\{\R^1,\R^2,\ldots,\R^\ell\}$.
Lemma~\ref{lem:weak-uncross}
 shows that for every $\R^i$,
$\mu(\R^i,\aP)\le   \mu(\Q,\aP)$ (note that $\P$ is
strongly cross-free with all members of $\aP$). Moreover, the
inequality must be strict, since $\Q$ and $\P$ are weakly cross-free,
but $\R^i$ and $\P$ are strongly cross-free.
Furthermore, $\nu(\R^i,\dunion\aP)=0$. This is because $\R^i$ consists
of sets in $\P\cup\Q$, and $\nu(\Q,\dunion\aP)=0$ implies that
$(\cup\aP)\cup \Q$ is a cross-free family.
Again, we have
\[
\chi(\P)+\chi(\Q)=\sum_{i=1}^\ell\chi(\R^i),
\]
by Lemma~\ref{lem:weak-uncross}(iii),
and therefore $\chi(\R^r)\notin \mbox{span}(\aP)$ for some $1\le
r\le \ell$.  Then
 $\mu(\R^r,\aP)<  \mu(\Q,\aP)$ contradicts the choice of $\Q$.
This completes the proof of Theorem~\ref{thm:charac}. \qquad \endproof

It now remains to prove Lemma~\ref{lem:weak-uncross} that we do in the next section.

\subsection{Uncrossing and Strongly Cross-free Families}\label{sec:uncrossing}

We first define the uncrossing operation $\Upsilon(.,.)$ that generalizes the intersection and union operation over sets to partitions and co-partitions. As an important difference, whereas uncrossing two sets results in two other sets, uncrossing a partition and a co-partition may lead to a collection of more than two partitions and co-partition.

 Assume $\P$ and $\Q$ are weakly cross-free partitions or co-partitions with  associated sub-partitions $\tilde\P$ and $\tilde\Q$.
By Claim~\ref{cl:sup-lam}, $\supp(\tilde\P)$ and $\supp(\tilde\Q)$ are
either disjoint, or one is a subset of the other. If they are disjoint,
then $\P$ and $\Q$ are strongly cross-free by definition; hence the
second alternative must hold. Further, if $\P$ is a partition and $\Q$
is a co-partition, then $\P\cup\bar\Q$ must be a laminar family. Indeed, if there were parts $P\in\P$, $\bar Q\in\bar\Q$ with $P\cup\bar Q=V$, then $\P$ and $\Q$ would be strongly cross-free by Claim~\ref{cl:strongly-cf}.
We are ready to define
 the set $\Upsilon(\P,\Q)$ that corresponds to the ``uncrossing'' of the weakly cross-free pair $\P$ and $\Q$.

\begin{itemize}
 \item {\em If $\P$ and $\Q$ are both partitions or both co-partitions} (see Figure~\ref{fig:uncrossing1}). Without loss of generality, let us assume $\supp(\tilde\P)\subseteq \supp(\tilde\Q)$.
 The set family $\tilde\P\dunion\tilde\Q$ is laminar and covers every
 element of $\supp(\tilde\P)$ exactly twice, and every element of $\supp(\tilde\Q)\setminus \supp(\tilde\P)$ exactly once. It is easy to see that it can be decomposed into two sub-partitions $\F$ and $\F'$ such that $\supp(\F)=\supp(\tilde\P)$, $\supp(\F')=\supp(\tilde\Q)$,
and $\F'$ dominates $\F$.
\begin{itemize}
\item If both $\P$ and $\Q$ are partitions, then
let us define
 $\P\wedge \Q:=\{P_1\}\cup \F$, and $\P\vee\Q:=\{Q_1\}\cup \F'$.
\item If both $\P$ and $\Q$ are co-partitions, then 
let us define $\P\wedge \Q:=\{P_1\}\cup \bar\F$, and $\P\vee\Q:=\{Q_1\}\cup \bar\F'$.
\end{itemize}
 In both cases, we let $\Upsilon(\P,\Q):=\{\P\wedge\Q,\P\vee\Q\}$.
Note that if both $\P$ and $\Q$ are partitions, then both $\P\wedge\Q$ and $\P\vee\Q$ are partitions; whereas if both $\P$ and $\Q$ are co-partitions, then they are both co-partitions as well.
\item {\em If $\P$ is a partition and $\Q$ is a co-partition} (see Figure~\ref{fig:uncrossing2}). Let
  $\P=(P_1,\ldots,P_p)$ and $\Q=(Q_1,\ldots,Q_q)$, with complement
  $\bar \Q$. Let $\F$ denote the maximal members of the laminar family
  $\P\cup\bar\Q$ (note that here we have $\cup$ instead of
  $\dunion$; the laminarity of $\F$ was verified above).
Let $\Upsilon(\P,\Q)$ denote the collection of the following $|\F|$
partitions and co-partitions.
\begin{itemize}
\item
For every  $\bar Q_j\in \F$ that is not a part of $\P$, define the
partition $\R$ consisting of $Q_j$ and the sets $P_i$ such that
$P_i\subseteq \bar Q_j$.
\item For every $P_i\in \F$ that is not a part
of $\bar \Q$, define the co-partition $\R$ consisting of $P_i$
and the sets $Q_j$ such that $\bar Q_j\subseteq P_i$.
\item For every $S\in\F$ that is both a part of $\P$ and $\bar \Q$,
  define
the partition $\R=\{S,V\setminus S\}$.
\end{itemize}
Since $\P$ and $\Q$ are weakly cross-free, the set $\F$ has a single member that contains the root \r. Let us call the corresponding partition or co-partition the {\em special} member of
$\Upsilon(\P,\Q)$.
\end{itemize}
In Figure~\ref{fig:uncrossing2}, $\R_1$, $\R_2$ and $\R_3$
illustrate the three respective cases above; the special member of
$\Upsilon(\P,\Q)$ is $\R_1$.

\medskip
We are ready to prove Lemma~\ref{lem:weak-uncross}.

{\em Proof of Lemma~\ref{lem:weak-uncross}.  }
For {\em(i)}, observe that  $\P\dunion\Q=\dunion\{\R': \R'\in
\Upsilon(\P,\Q)\}$ in each of the cases.
Then the last part of Claim~\ref{cl:tight-fam} together with $\Psi_x(\P\dunion\Q)=0$ implies that all
members of $\Upsilon(\P,\Q)$ are tight.
Parts {\em(ii)} and {\em(iii)} are straightforward to check.

For {\em(iv)}, first observe that $\S\cup \R$ is always cross-free,
since $\R$ consists of certain parts of $\P\cup\Q$. We verify the additional properties of strongly cross-freeness in the different cases.
\begin{enumerate}
\item {\em $\P$, $\Q$ are both partitions or both co-partitions.} Again by Claim~\ref{cl:sup-lam}, $\supp(\tilde\P)$ and $\supp(\tilde\Q)$ are either disjoint or one is a subset of the other; and they cannot be disjoint since we assumed that $\P$ and $\Q$ are weakly cross-free. Without loss of generality, let us assume $\supp(\tilde\P)\subseteq \supp(\tilde\Q)$ as in Figure~\ref{fig:uncrossing1}.
\begin{enumerate}
\item $\R=\P\wedge \Q$. By definition, $\supp(\tilde\R)=\supp(\tilde\P)$; also note that both $\tilde\P$ and $\tilde\Q$ dominate $\tilde\R$.
 We are done if $\tilde\S$ and $\tilde\P$ are disjoint; otherwise, one of them dominates the other.
\begin{claim}
Assume $\supp(\tilde\S)\cap\supp(\tilde\P)\neq\emptyset$.
If $\tilde\S$ dominates at least one of $\tilde\P$ or $\tilde\Q$, then it also dominates $\tilde\R$. Otherwise, $\tilde\R$ dominates $\tilde\S$.
\end{claim}
\begin{proof}
If $\tilde\S$ dominates at least one of $\tilde\P$ or $\tilde\Q$, then it also dominates $\tilde\R$ by the transitivity of domination (Claim~\ref{cl:trans}).
Assume this is not the case, and hence  both $\tilde\P$ and $\tilde\Q$ dominate $\tilde\S$. (Notice that $\tilde\S$ and $\tilde\Q$ cannot be disjoint because of $\supp(\tilde\P)\subseteq\supp(\tilde\Q)$).
We show that $\tilde\R$ dominates $\tilde\S$. Let $S\in\tilde\S$ be an arbitrary part; by definition, there exists  parts $P\in\tilde\P$ and $Q\in\tilde\Q$ with $S\subseteq P$, $S\subseteq Q$. Since $\tilde\P\cup\tilde\Q$ is laminar (Claim~\ref{cl:sup-lam}), we have either $P\subseteq Q$ or $Q\subseteq P$. The smaller one among the two sets is contained in $\tilde\R$, proving the claim.
\end{proof}

To complete the proof of strong cross-freeness, we have to consider the case when one of $\S$ and $\R$ is a partition and the other is a co-partition.
Assume both $\P$ and $\Q$ and thus $\R$ are partitions, and $\S$ is a co-partition.
Now $\tilde\S$ is disjoint from $\tilde\P$, or one of them strongly dominates the other; the same applies for $\tilde\S$ and $\tilde\Q$. We are done if $\tilde\S$ is disjoint from either of them.
If $\tilde\S$ strongly dominates $\tilde\P$ or $\tilde\Q$, then the second part of Claim~\ref{cl:trans} implies that it also strongly dominates $\tilde\R$.
If $\tilde\S$ is strongly dominated by both $\tilde\P$ and $\tilde\Q$, then there must be parts $P\in \tilde\P$, $Q\in \tilde\Q$ with $\supp(\tilde\S)\subseteq P\cap Q$. As in the proof of the above claim, $P\subseteq Q$ or $Q\subseteq P$, and the smaller one must be a part of $\tilde\R$, and therefore $\tilde\R$ strongly dominates $\tilde\S$.
The same arguments work whenever $\P$, $\Q$ and $\R$ are co-partitions, and $\S$ is a partition.
\item $\R=\P\vee \Q$. A similar argument is applicable. Now
  $\supp(\tilde\R)=\supp(\tilde\Q)\supseteq \supp(\tilde\P)$, and
  $\tilde\R$ dominates both $\tilde\P$ and $\tilde\Q$. We are done if
  $\tilde\Q$ and $\tilde\S$ are disjoint. As in the above proof, one
  can verify that if either of $\tilde\P$ and $\tilde\Q$ dominates
  $\tilde\S$, then $\tilde\R$ also dominates $\tilde\S$ by
  transitivity. If $\tilde\S$ dominates both $\tilde\P$ and
  $\tilde\Q$, then it also must dominate $\tilde\R$. There is one more
  case however, when  $\supp(\tilde\S)\subseteq
  \supp(\tilde\Q)\setminus \supp(\tilde\P)$. But in this case,
  $\tilde\Q$ must dominate $\tilde\S$, a case we have already covered.

Consider the case when one of $\S$ and $\R$ is a partition and the other is
a co-partition.
 If $\tilde\Q$ and $\tilde\S$
are disjoint, then $\tilde\R$ and $\tilde\S$ are also disjoint since $\supp(\tilde\R)=\supp(\tilde\Q)$.
If $\tilde\Q$ strongly
dominates $\tilde\S$, then $\tilde\R$ also strongly dominates
$\tilde\S$ using Claim~\ref{cl:trans}. Finally, if $\tilde\S$ strongly
dominates $\tilde\Q$, then for some part $S\in\tilde\S$ we have
$\supp(\tilde\Q)\subseteq S$, implying that $\tilde\S$ strongly
dominates $\tilde\R$ again by  $\supp(\tilde\R)=\supp(\tilde\Q)$.
\end{enumerate}

\item {\em $\P$ is a partition and $\Q$ is a co-partition.} We refer to Figure~\ref{fig:uncrossing2} for an example. By
  symmetry, we may assume that $\S$ is also a partition; the proof
  extends to the case when $\S$ is a co-partition by swapping the
  roles of $\P$ and
  $\Q$.
First, let us examine the case when $\R$ is a {\em co-partition}, that
is, for some $P_j\in\P$,  such that $P_j$ is a maximal member of
$\P\cup\bar\Q$, we have $\R=\{P_j\}\cup \{Q\in \Q:\bar Q\subseteq
P_j\}$. According to Claim~\ref{cl:strongly-cf}, there exists parts
$Q\in \Q$ and $S\in\S$ with $Q\subseteq S$. If $\bar Q\subseteq P_j$,
then $Q$ is also a part of
$\R$, and therefore $\S$ and $\R$ are strongly cross-free again by
Claim~\ref{cl:strongly-cf}.
Otherwise, $\bar Q\cap P_j=\emptyset$, and thus $P_j\subseteq
Q\subseteq S$, and hence $P_j\subseteq S$ verifies that $\S$ and $\R$
are strongly cross-free.

For the rest of the proof, we assume that $\R$ is a partition, that
is, for some $Q_j\in\Q$,  such that $\bar Q_j$ is a maximal member of
$\P\cup\bar\Q$, we have $\R=\{Q_j\}\cup \{P\in \P:P\subseteq
\bar{Q_j}\}$. We distinguish two cases, depending on whether $\R$ is
the special member of $\Upsilon(\P,\Q)$. Recall that the special member of $\Upsilon(\P,\Q)$ is defined by the maximal set in $\P\cup \Q$ that contains the root $r$.

\begin{enumerate}
\item{\em  $\R$ is {\em not} the special member of $\Upsilon(\P,\Q)$,
  that is, $j>1$, or equivalently, $r\notin\bar{Q_j}$.} It is easy to see that both $\P$ and $\Q$ then dominate $\R$.
Note that $\supp(\tilde\R)=\bar{Q_j}$. Since both $\S$ and $\R$ are partitions, it is sufficient to verify that either $\tilde\S$ and $\tilde\R$ are disjoint
or one dominates the other. Consider the sub-partitions $\tilde\S$ and $\tilde\Q$.
If they are disjoint, then so are $\tilde\S$ and $\tilde\R$. If $\tilde\S$ strongly dominates $\tilde \Q$, then it also dominates $\tilde \R$ by transitivity. The remaining case is when $\tilde\Q$ strongly dominates $\tilde\S$, that is, $\supp(\tilde\S)\subseteq Q'$ for some $Q'\in\tilde\Q$. If $Q'\neq \bar{Q_j}$ then $\tilde\S$ and $\tilde\R$ are disjoint; hence we may assume that $Q'=\bar{Q_j}$.

Consider now the relation between $\tilde\S$ and $\tilde\P$. Again, if $\tilde\S$ dominates $\tilde\P$, then it also dominates $\tilde\R$ by transitivity and hence we are done; also $\tilde\S$ and $\tilde\P$ cannot be disjoint, since $\bar{Q_j}\subseteq\supp(\tilde\P)$. Hence $\tilde\P$ dominates $\tilde\S$, that is, for every $S\in\tilde\S$, there exists a part $P\in\tilde\P$ with $S\subseteq P$. Because of
$\supp(\tilde\S)\subseteq\bar{Q_j}$ we must have $P\subseteq\bar{Q_j}$ and thus $P$ is also a part of $\R$, showing that $\R$ dominates $\S$.

\item{\em  $\R$ is the special member of $\Upsilon(\P,\Q)$,
  that is, $j=1$, or equivalently, $\r\in\bar{Q_j}$.} It is easy to see that $\R$ dominates both $\P$ and $\Q$.
Note that $P_1\in\R$ and $\supp(\tilde\R)=\bar P_1=\supp(\tilde\P)\supseteq \supp(\tilde\Q)$.
Let us consider $\tilde\S$ and $\tilde\P$. If they are disjoint, then so are $\tilde\S$ and $\tilde\R$. If $\tilde\P$ dominates $\tilde\S$, then $\tilde\R$ also dominates $\tilde\S$ by transitivity. Hence we may assume that $\tilde\S$ dominates $\tilde\P$.

This implies that $\supp(\tilde\Q)\subseteq\supp(\tilde\P)\subseteq\supp(\tilde\S)$ and therefore $\tilde\Q$ and $\tilde\S$ cannot be disjoint. If $\tilde\Q$ strongly dominates $\tilde\S$, then we would again have that $\tilde\R$ dominates $\tilde\S$.
Therefore $\tilde\S$ strongly dominates $\tilde\Q$, that is, $Q_1=\supp(\tilde\Q)\subseteq S$ for some $S\in\tilde\S$. Together with the fact that $\tilde\S$ dominates $\tilde\P$, this implies that $\tilde\S$  dominates $\tilde\R$. Indeed, all parts of $\tilde\R$ different from $Q_1$ are also parts of $\tilde\P$ and are hence subsets of some parts of $\tilde\S.$
\end{enumerate}
\end{enumerate}

This completes the proof of Lemma~\ref{lem:weak-uncross}. $\qquad \endproof$

\subsection{The partial order}\label{sec:partial}
The following notion of a partial order and the properties we derive will be crucial in the proof
of Theorem~\ref{thm:frac-value}.
Consider a strongly cross-free family $\PP$ of partitions and co-partitions.
We can naturally define a partial ordering $(\PP,\preceq)$ as follows. For $\P,\Q\in\PP$, let $\P\preceq\Q$ if $\tilde\Q$ dominates $\tilde\P$. Claim~\ref{cl:trans} immediately implies that $\preceq$ is a partial order. Since $\PP$ is strongly cross-free, if $\P$ and $\Q$ are incomparable then $\tilde\P$ and $\tilde\Q$ are disjoint.
Note that if one of $\P$ and $\Q$ is a partition and the other is a co-partition, then
$\P\preceq\Q$ if and only if $\tilde\Q$ strongly dominates $\tilde\P$.

Let us say that a partition or co-partition $\P$ {\em properly contains} a set
$S\subseteq V$ if $S\subseteq P$ for some part $P$ of the associated
sub-partition $\tilde\P$.

\begin{lemma}\label{lem:proper-contain}
Let $S\subset V$ and let $\P$ and $\Q$ be two strongly cross-free partitions/co-partitions properly containing $S$. Then either $\P\preceq\Q$ or $\Q\preceq\P$. Moreover, if $\P$ properly contains $S$, then every $\Q$ with
$\P\preceq\Q$ also properly contains $S$.
\end{lemma}
 \begin{proof}
If $\P$ and $\Q$ both properly contain $S$ then $\supp(\tilde\P)\cap\supp(\tilde\Q)\neq\emptyset$, and therefore $\P\preceq\Q$ or $\Q\preceq\P$ due to the strongly cross-free assumption.
For the second property, let $\P\in \PP$ be properly containing $S$ and let $\P\preceq \Q$. Then $S\subseteq P$ for some part $P\in\tilde\P$ by the proper containment,
and $P\subseteq Q$ for some $Q\in\tilde\Q$ by domination, thus $S\subseteq Q$, as required.
\end{proof}

An easy consequence is the following lemma, that enables tree representations of strongly cross-free families.
\begin{lemma}\label{lem:tree}
Let $\PP$ be a strongly cross-free family of partitions and co-partitions. Then there is a rooted forest $F=(\PP,E(F))$ such that $\Q$ is an ancestor of $\P$ in $F$ if and only if $\P\preceq\Q$.
\end{lemma}
\begin{proof}
Lemma~\ref{lem:proper-contain} implies that if $\P,\Q,\R\in \PP$ such that $\P\preceq\Q$ and $\P \preceq\R$ then either $\Q\preceq\R$ or $\R\preceq\Q$. Indeed, let $S=\{v\}$ for an arbitrary $v\in\supp(\tilde\P)$, and apply Lemma~\ref{lem:proper-contain} for $\Q$ and $\R$.
Now, we construct the forest $F$ as follows. We consider all the maximal partitions or co-partitions under the partial order $(\PP,\preceq)$ as roots. For every $\P\in \PP$ which is not maximal, we let its parent be $\Q\in \PP$  that is the $\preceq$-minimal partition/co-partition in $\PP$ such that $\P\preceq\Q$. The above claim implies that $\Q$ is uniquely defined and therefore, we obtain a forest.
\end{proof}

Our next lemma characterizes the relation between two strongly cross-free partitions or co-partitions.

\begin{lemma}\label{lem:type-1}
Assume that $\P\preceq\Q$ for $\P,\Q\in\PP$ both partitions or both co-partitions.
Then either $\tilde\Q$ strongly dominates $\tilde\P$, that is, $\supp(\tilde\P)\subseteq Q$ for some part $Q\in \tilde\Q$, or $\tilde\P$ contains a partition of every part of $\tilde\Q$ that intersects $\supp(\tilde\P)$.
\end{lemma}

\begin{proof}
Assume $\Q$ dominates but not strongly dominates $\P$; let $Q\in\tilde\Q$ be a part with $Q\cap\supp(\tilde\P)\neq\emptyset$.
We will show that $\tilde\P$ contains a partition of $Q$.
Observe that $Q\cup \supp(\tilde\P)\neq V$ as neither of them contains the root $\r$.
 By the assumption, $\tilde\Q$ must also have another part, say $Q'$ with $Q'\cap\supp(\tilde\P)\neq\emptyset$; this shows $\supp(\tilde\P)\setminus Q\neq\emptyset$.
Since $\P\cup\Q$ is cross-free, $Q$ and $\supp(\tilde\P)$ must not cross.
Therefore we may conclude that $Q\subseteq\supp(\tilde\P)$.

Consider now a part $P\in\tilde\P$ with $P\cap Q\neq\emptyset$.
By domination, there exists a $Q'\in \tilde\Q$ with $P\subseteq Q'$; since $\tilde\Q$ is a sub-partition, we must have $Q'=Q$. Hence every part of $\tilde\P$ intersecting $Q$ is entirely contained inside $Q$, showing that $\tilde\P$ contains a partition of $Q$.
\end{proof}

\subsection{Proof of Theorem~\ref{thm:frac-value}}\label{sec:iterative}

Let us assume that for each edge $e\in E^*$,
$0<x^*_e<\frac{1}{\ratio}$. Let $\aP$  be the family of partitions and
co-partitions as given by Theorem~\ref{thm:charac}. We will derive a
contradiction to the fact that $|E^*|=|\supp(x^*)|= |\aP|$ via
a counting argument. We start by assigning three tokens to
each edge $e\in E^*$ and reassign them to each partition/co-partition
in $\aP$  such that we assign three tokens to each of them, and have
some extra tokens left, a contradiction.

Before giving the token argument, we begin with a few definitions.
A  partition or co-partition  $\P$
\emph{properly covers}  an edge $(u,v)$ if
$u$ and $v$ lie in different parts of the sub-partition $\tilde\P$.
 A  partition or co-partition  $\P$
\emph{semi-covers} an edge $(u,v)$
 if exactly one of the sets $\{u\}$ or $\{v\}$ is properly contained by $\P$. Equivalently, the edge $e=(u,v)$ is semi-covered if $e\in \delta(P_1)$.
 We let $I(\P)$ denote the set of edges that are properly covered by
 $\P$ and let $i(\P)\defeq |I(\P)|$.

Recall from
Lemma~\ref{lem:proper-contain} that for any set $S\subseteq V$ if there exists a partition or co-partition in $\PP$ properly containing $S$, then there
is a $\preceq$-minimal $\P\in \PP$ such that $\P$ properly contains
$S$. The tokens are initially assigned according to two rules.

\medskip
\noindent\textbf{Initial Assignment I.} Each edge $(u,v)$ distributes
its three tokens as follows. For each of the two endpoints of $e$, say $u$, we assign one token to
the $\preceq$-minimal  $\P\in\PP$
such that $\P$ properly contains $\{u\}$. We also assign one token to the $\preceq$-minimal $\P\in \PP$ such that $\P$ properly contains the set $\{u,v\}$.

\begin{lemma}
Let $\P$ be a leaf in the forest $F$. Then it receives at least
$7+i(\P)$ tokens from Initial Assignment I.
\end{lemma}
\begin{proof}
It is easy to see that $\P$ receives at least one token for each edge in $\chi(\P)$ and two tokens for each edge that is properly covered by $\P$ for a total of at least $|\chi(\P)|+i(\P)$. Since $x(\chi(\P))\geq 1$, and $x_e< \frac{1}{\ratio}$, we may conclude that $|\chi(\P)|\geq 7$, implying the claim.
\end{proof}

\noindent
\textbf{Initial Assignment II.} We collect one token from each leaf in $F$ and give one token to each branching node (a node with at least two children) in $F$.

This assignment is feasible since the number of branching nodes in a
forest is at most the number of leaves. At the end of this initial
assignment, each of the leaves has at least $6+i(\P)$ tokens and
branching nodes have at least one tokens from the leaves. We now
proceed by induction.

The following lemma
completes the proof of Theorem~\ref{thm:frac-value}: if applied
it to the root nodes of the forest $F$, we can see that every member
of $\PP$ is assigned 3 tokens and we have a positive amount of surplus
tokens left. This contradicts $|E^*|= |\aP|$.

\begin{lemma}\label{lem:token-induction}
Let $\Q$ be any node in $F$, and let $F'$ denote the subtree rooted at
$\Q$. The tokens assigned to the nodes of $F'$ by Initial Assignments I and II
can be reassigned in such a way that every node in $F'$  receives three tokens
 and $\Q$ receives at least $6+\min\{1,i(\Q)\}$ tokens.
\end{lemma}

\begin{proof}
The proof of the lemma is by induction on the height of the subtree. Clearly, the Initial Assignment II ensures that there are at least $6+i(\Q)$ tokens assigned to leaves and hence the base case of the induction holds.

Now consider an internal node $\Q$. First suppose that it is a
branching node and thus has at least two children. Applying induction
on the subtree rooted at each of the children, we obtain that every
child $\P$ can be assigned at least $6+\min\{1,i(\P)\}$ tokens, and every other node in the subtree obtains at least three tokens. We keep the same assignment for all the nodes in each of the subtrees. Each child gives three tokens to $\Q$ while keeping at least three tokens for itself. Thus $\Q$ receives at least $6$ tokens from its children. Moreover, it receives at least one token from the leaves in the Initial Assignment II. Thus it receives at least $7\geq 6+\min\{1,i(\Q)\}$ tokens completing the induction.

Now consider the case when $\Q$ has exactly one child, say $\P$. We first apply the inductive hypothesis on the subtree rooted at $\P$. Thus $\P$ can donate $3+\min\{1,i(\P)\}$ tokens to $\Q$, while maintaining that all nodes in the subtree of $\P$, including $\P$, receive at least three tokens. We now show that $\Q$ receives the remaining required tokens by Initial Assignment I. We consider two different cases, depending on whether $\tilde\Q$ strongly dominates $\tilde\P$ or not. Note that if one of $\P$ and $\Q$ is a partition and the other is a co-partition,  then $\tilde\Q$ must strongly dominate $\tilde\P$ by the definition of strongly cross-freeness. However, if they are both partitions or both co-partitions, only the weaker property of domination is assumed. Throughout the arguments, we shall use that since
$\P$ is the only child of $\Q$, every $\R\prec\Q$ must satisfy
$\supp(\tilde\R)\subseteq\supp(\tilde\P)$. Consequently, if an edge $(u,v)$ has an endpoint in $\supp(\tilde\Q)\setminus\supp(\tilde\P)$, then $\Q$ must receive the corresponding token in Initial Assignment I.

\begin{enumerate}
\item {\em $\tilde\Q$ does not strongly dominate $\tilde\P$}; recall the structure described in Lemma~\ref{lem:type-1}.

\begin{claim}\label{cl:type-1}
$\Q$ receives at least one token from every edge in
$\chi(\P)\Delta\chi(\Q)$ in Initial Assignment I.
\end{claim}
\begin{proof}
First consider any edge $(u,v)\in\chi(\Q)\setminus \chi(\P)$.
By Lemma~\ref{lem:type-1}, $\tilde\P$ contains a partition of every part of $\tilde\Q$ that intersects $\supp(\tilde\P)$. Therefore we must have $\{u,v\}\cap\supp(\tilde\P)=\emptyset$; since $(u,v)\in\chi(\Q)$, $\Q$ must properly contain at least one of  $\{u\}$ and $\{v\}$. In either case, $\Q$ is the $\preceq$-minimal set with this property and receives at least one token for each such edge by Assignment I.

Now consider an $(u,v)\in \chi(\P)\setminus \chi(\Q)$. It must be the case that $u$ and $v$ are in different parts of $\tilde\P$ and moreover, both $u$ and $v$ are contained in a single part of $\tilde\Q$. Thus $\Q$ is the $\preceq$-minimal set properly containing $\{u,v\}$, and therefore receives one token by Assignment I for this edge.
\end{proof}

We now show that these tokens are sufficient for $\Q$ using the following claim.
\begin{claim}\label{cl:delta}
$x^*(\chi(\P)\Delta\chi(\Q))$
is  a positive integer.
\end{claim}
\begin{proof}
Let us assume that both $\P$ and $\Q$ are partitions; the proof is the same for the co-partition case.
We can give an alternative decomposition of the 2-regular family $\P\dunion\Q$ into two  partitions as follows.
    Let the partition $\R$ consist of the minimal sets in $\P\cup \Q$ and let $\S$ be the partition consisting of the maximal sets in $\P\cup \Q$. Now $\P\dunion\Q=\R\dunion\S$, and therefore
Claim~\ref{cl:tight-fam} implies that both $\R$ and $\S$ are tight. Thus $x^*(\chi(\R))-x^*(\chi(\S))$ equals an integer. But $\chi(S)\subseteq \chi(R)$ and $\chi(R)\setminus \chi(S)=\chi(\P)\Delta\chi(\Q)$.
Finally, $x^*(\chi(\P)\Delta\chi(\Q))=0$   would contradict the linear independence assumption in Theorem~\ref{thm:charac}.
\end{proof}

Since $x^*_e<\frac1\ratio$ for every edge, we may conclude from the above Claim that $|\chi(\P)\Delta\chi(\Q)|\geq 7$ and therefore $\P$ receives at least $3+\min\{1,i(\P)\}+7> 6+\min\{1,i(\Q)\}$ tokens.

\item {\em $\tilde\Q$ strongly dominates $\tilde\P$}; assume that part $Q$ of $\tilde\Q$ contains $\supp(\tilde\P)$.

\begin{claim}
$\Q$ receives at least $|\chi(\Q)\Delta \chi(\P)|+i(\Q)$ tokens in Initial Assignment I. In case of  equality, either
$\chi(\P)\subseteq \chi(\Q)$ or
$i(\P)\ge 1$ must hold.
\end{claim}
\begin{proof}
For every edge $(u,v)\in\chi(\Q)\setminus \chi(\P)$ we must have $\{u,v\}\cap \supp(\tilde\P)=\emptyset$.
In Initial Assignment I, $\Q$ must receive  $|\{u,v\}\cap \supp(\tilde\Q))|$ tokens from the edge $(u,v)$; that is, it receives one if $\Q$ semi-covers $(u,v)$ and 2 if
$(u,v)\in I(\Q)$.  Altogether $\Q$ receives at least $|\chi(\Q)\setminus \chi(\P)|+|I(\Q)\cap(\chi(\Q)\setminus \chi(\P))|$ tokens from these edges. Meanwhile if there is an edge $(u,v)\in I(\Q)\cap \chi(\P)$, then one of its endpoints must be in a part of $\tilde\Q$ different from $Q$. Thus, $\Q$ is the $\preceq$-minimal (co-)partition properly  containing this endpoint and therefore receives one token for this edge.
In summary, $\Q$  receives at least $|\chi(\Q)\setminus \chi(\P)|+i(\Q)$ tokens from edges in $\chi(\Q)$.

Consider an edge $(u,v)\in \chi(\P)\setminus \chi(\Q)$. We must have  $\{u,v\}\subseteq Q$, and therefore $\Q$ is the $\preceq$-minimal (co-)partition properly containing $\{u,v\}$, and receives at least one token for any such edge.
Hence $\Q$ receives the claimed total token amount.

Let us now focus on the second part of the claim; assume $\chi(\P)\setminus \chi(\Q)\neq\emptyset$, and let $(u,v)\in \chi(\P)\setminus \chi(\Q)$.
If $(u,v)$ semi-covers $\P$, then one of the endpoints must lie in $Q\setminus \supp(\P)$ and hence $\Q$ receives an extra token corresponding to this endpoint.
Therefore in the case of equality, every $(u,v)$ must properly cover $\P$, that is, $i(\P)\ge 1$.
\end{proof}

Subtracting the constraint for $\P$ from that for $\Q$, we obtain that
$x^*(\chi(\Q))-x^*(\chi(\P))=t\in\ZZ$.
Observe that if $t\neq 0$, then $|\chi(\Q)\Delta \chi(\P)|\geq 7$ due to the assumption $x^*_e<\frac1\ratio$ for every edge. In this case, we are done by the above Claim.

The rest of the argument focuses on the case $t=0$. We must have $\chi(\P)\setminus \chi(\Q)$ and $\chi(\Q)\setminus \chi(\P)$  both non-empty due to the linear independence.
Together with the tokens $\Q$ can obtain from $\P$, it  receives a total of at least
$$3+\min\{1,i(\P)\}+ 2+i(\Q).$$
Further, equality may only hold if $i(\P)\geq 1$. Thus, whether $i(\P)=0$ or not, $\Q$ receives a total of at least
$6+\min\{1,i(\Q)\}$ tokens, completing the induction.
\end{enumerate}
\end{proof}
This completes the proof of Theorem~\ref{thm:frac-value}.

\section{Integrality Gap}\label{sec:mixed}

In this section we prove Theorem~\ref{thm:orientation}, giving an
integrality gap example for the
\emph{Augmenting a mixed graph with orientation constraints problem}. We can formulate the following natural linear
program.
This is identical to \eqref{lp:2} with the requirement function
$f(Z)=k-d^{in}_A(Z)$. This function $f$ is indeed crossing
$(G-)$supermodular; however, note that it might take negative values
as well.

\begin{equation}\tag{{\em
        LP3}} \label{lp:3}
\begin{aligned}
    \mbox{ minimize} \sum_{(u,v)\in E^*} c_{uv} x_{uv} & &\notag \\
    \mbox{ s.t.}~~
     y(\delta^{in}(Z))& \geq k-d^{in}_A(Z)  &  \;\; \forall\; \emptyset\neq Z\subsetneq V\notag\\
     y_{uv}+y_{vu} &= x_{uv}& \;\; \forall\; (u,v)\in E^* \\
      y_{uv}+y_{vu}&= 1 &\;\; \forall (u,v)\in E\notag\\
 \textbf{0}\le x&\le \textbf{1} &\notag\\
y&\ge \textbf{0}
\end{aligned}
\end{equation}

We present an
example for $k=2$ with integrality gap $\Omega(|V|)$. Our example can
 be easily extended for an arbitrary value $k>2$, see Remark~\ref{rem:big-k}.
The same construction will be used to show that \eqref{lp:2} cannot be used for iterative rounding even
for nonnegative demand functions, see Remark~\ref{rem:lp:2},

\begin{figure*}[htb]
\centering
\mbox{\subfigure{\includegraphics[width=4cm,height=7cm]{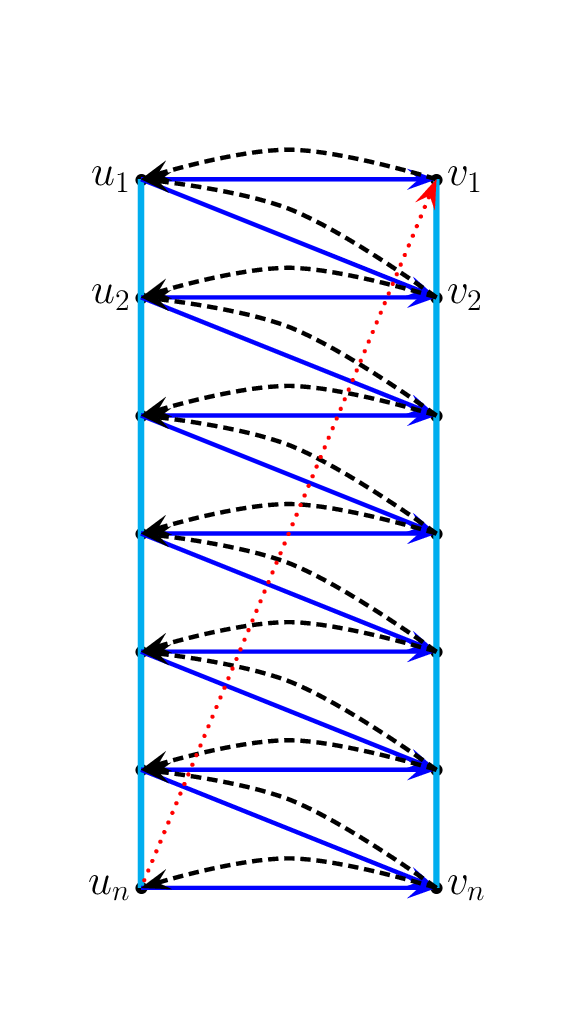}}\quad
\subfigure{\includegraphics[width=8cm,height=7cm]{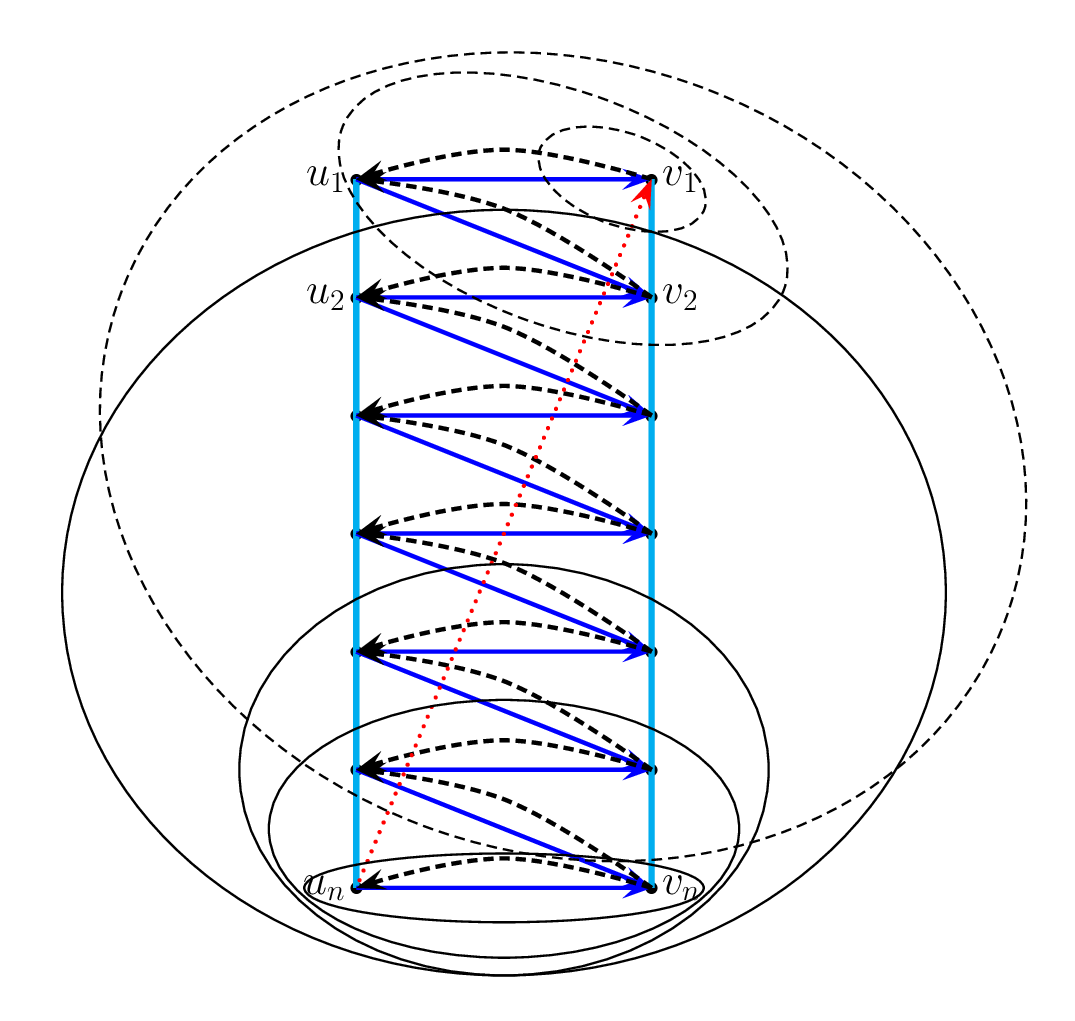} }}
\caption{Integrality gap example and the set of fundamental cuts. The special arc $e_r=(u_n,v_1)$ costs 1. The  arcs $A_G=\{(u_i,v_i):1\leq i  \leq n\} \cup \{(u_i,v_{i+1}):1\leq i
\leq n-1\}$ are the solid horizontal and diagonal arcs. The arcs $A_B$ consists of the dashed arcs and include two copies in the reverse direction for each arc in $A_G$. The set $E= \{\{u_i,u_{i+1}\}: 1\leq i\leq n-1\}\cup \{\{v_i,v_{i+1}\}: 1\leq
i\leq n-1\}$ include all the vertical arcs.  The $2n-1$ fundamental cuts in the second figure certify that there does not exist an orientation of the mixed graph without the special arc $e_r$ that is $2$-edge connected.} \label{fig:1}
\end{figure*}

\subsubsection*{Integrality Gap Instance.}
The graph (see Figure~\ref{fig:1}) is on $2n$ nodes:
$V=\{u_1,\ldots,u_n,v_1,\ldots,v_n\}$.
Let us first define the set of arcs $A=A_G\cup A_B$.
\[
A_G\defeq\{(u_i,v_i):1\leq i  \leq n\} \cup \{(u_i,v_{i+1}):1\leq i
\leq n-1\}.
\]
The second group, $A_B$ consists of two parallel copies of the reverse
of every edge
in $A_G$ (each curved arc represent two arcs).
The set $E$ of edges already available is
\[
E\defeq\{\{u_i,u_{i+1}\}: 1\leq i\leq n-1\}\cup \{\{v_i,v_{i+1}\}: 1\leq
i\leq n-1\}.
\]
The set $E^*$ consists of a single edge  $e_r=(u_n,v_1)$ with $c(e_r)=1$.
Let us first examine integer optimal solutions. We show that every
2-edge-connected oriented subgraph must use the arc $(u_n,v_1)$ and hence
must have cost at least 1.
\begin{lemma}
The optimal integral solution to \eqref{lp:3}
equals 1.
\end{lemma}
\begin{proof}
We show that there
is no possible orientation of the edges of $E$ which offers a
feasible solution together with $A$, and thus one must pick the edge $e_r$.
We then show that there is a feasible solution after picking $e_r$.

Consider the directed tree induced by $A_G$. For every arc in $A_G$, we
consider the fundamental cut
entered by this arc. Namely, for every arc $(u_i,v_i)$ let
$(\bar S_i,S_i)$ be the cut with $S_i=\{u_1,\ldots,u_{i-1},
v_{1},\ldots,v_i\}$ for each $1\leq i\leq n$.
For every arc $(u_i,v_{i+1})$,  let $(\bar T_i,T_i)$
be the cut with $T_i=\{u_{i+1},\ldots,u_n,v_{i+1},\ldots,v_n\}$ for each $1\leq i
\leq n-1$.  Let $\F$ denote the family of these $2n-1$ cuts.
Since, the final solution is $2$-strongly connected, there
must be at least two arcs entering the sets $S_i$ and $T_i$ for each $i$. Thus the total demand of the cuts in $\F$ is
exactly $4n-2$.

The arcs in $A_B$ do not enter any of these cuts, and  each arc in $A_G$
enters exactly one of them. Hence
$A=A_G\cup A_B$ supplies exactly  $2n-1$ arcs covering $\F$.
 Moreover, each of the $2n-2$ edges in $E$ may enter exactly one cut in
 the family, whichever direction it is oriented in. For example, the
 edge $(u_i,u_{i+1})$ enters $T_i$ if oriented from $u_i$
 to $u_{i+1}$ and enters $S_{i}$ if oriented from $u_{i+1}$
 to $u_i$.
This shows that $A$ together with any orientation of $E$ contains
exactly $4n-3<4n-2$ arcs covering cuts in $\F$, and thus there
cannot be any feasible solution of cost 0.

We now construct a cost 1 solution.
Let us orient $e_r$ from $u_n$ to $v_1$.
Orient the edge $(u_i,u_{i+1})\in E$ from $u_i$ to $u_{i+1}$, and orient
the edge $(v_i,v_{i+1})\in E$ from $v_i$ to $v_{i+1}$ for each $1\leq i
\leq n-1$.
We now show that this is a 2-edge-connected graph.
Consider any set $\emptyset\neq Z\subsetneq V$. If an arc in $A_G$ leaves
$Z$, then there are two arcs
in reverse direction in $A_B$ entering $Z$. We are also done if there exists two arcs $(u,v),
(u',v')\in A_G$, both entering $Z$.
Thus the only cuts that need to
checked are the set $\F$ of fundamental cuts; note that there is
always an arc in $A_G$ entering these sets. We claim that at least
one of the oriented arcs in $E$ or $e_r$ must also enter it. This follows since $e_r$ enters ${S}_i$ for each $1\leq i \leq n$. The arc $(u_i,u_{i+1})$ enters ${T}_i$ for each $1\leq i \leq n-1$ completing the proof.
\end{proof}

\begin{lemma}
There is fractional solution to \eqref{lp:3}  of cost $\frac{1}{n}$.
\end{lemma}
\begin{proof}
We give the following fractional solution.
Let
\[
\begin{array}{lll}
 y_{u_iu_{i+1}}=1-\frac{i}{n},&
y_{u_{i+1}u_{i}}=\frac{i}{n},& \\
 y_{v_iv_{i+1}}=\frac{i}{n},&
y_{v_{i+1}v_{i}}=1-\frac{i}{n},& \quad \forall 1\leq i \leq n-1,\\
x_{u_nv_1}=y_{u_nv_1}=\frac{1}{n},& y_{v_1u_n}=0.& .
\end{array}
\]
The cost is
 exactly $\frac{1}{n}$ due to the edge $e_r$.
We now check feasibility. As in the analysis of the integral solution,
it is sufficient to show that the fundamental cuts in $Z\in \cal F$ have a
total of at least 1 fractional edges entering them (note that the
demand is
$2-d^{in}_A(Z)=1$ for all $Z\in\F$).
 First consider the set $T_i$. The arcs entering it are exactly
$(u_i,u_{i+1})$ and $(v_i,v_{i+1})$. Thus the total fractional value entering this cut is exactly
$$y_{u_iu_{i+1}}+y_{v_iv_{i+1}}=1-\frac{i}{n}+\frac{i}{n}=1$$
as required.
Now consider the cut $S_i$. The arcs entering it are  $(u_{i},u_{i-1})$, $(v_{i+1},v_i)$ and $(u_n,v_1)$. The total fractional value entering this cut is exactly
$$y_{u_{i}u_{i-1}}+y_{v_{i+1}v_i}+y_{u_nv_1}=\frac{i-1}{n}+1-\frac{i}{n}+\frac{1}{n}=1$$
as required. Hence, the optimal value of LP has objective at most $\frac{1}{n}$.
\end{proof}

\begin{remark}\label{rem:big-k} To extend the integrality for any
  $k\geq 2$, we can add $k-2$ parallel directed Hamiltonian cycles
  $(u_1,\ldots,u_n,v_n,\ldots,v_1,u_1)$.
 A simple check shows that the total supply and
  demand of the fundamental cuts increases by exactly $(k-2)(2n-1)$,
  thus giving an integrality gap example for connectivity requirement
  of $k$.
\end{remark}

\begin{remark}\label{rem:lp:2}
It can also be shown that the fractional solution is an
  extreme point of \eqref{lp:3}. This can be used to show that
 \eqref{lp:2} for the Minimum Cost $f$-Orientable Subgraph Problem is not amenable for iterative rounding.
Indeed, in place of $E$ in \eqref{lp:2} we take the union of $E$ and the underlying
undirected edge set of $A$, and $E^*=\{e_r\}$; we set $x_{uv}=y_{uv}=1$
on the arcs in $A$. This then gives an extreme point of \eqref{lp:2}
with the single edge in $E^*$ having $x_{e_r}=\frac 1n$.
\end{remark}
\section{Discussion}\label{sec:discussion}

In this paper, we investigated two seemingly similar problem
settings.
 Our main result gave a 6-approximation for the {\em Minimum-Cost $f$-Orientable Subgraph Problem}, where $f$ is a nonnegative
valued crossing $G$-supermodular function, and $G$ denotes the
subgraph available for free. This includes the requirement of
$(k,\ell)$-edge-connectivity.
In the second setting, {\em Augmenting a Mixed Graph with Orientation Constraints} we aimed for the simpler requirement of global
$k$-connectivity, however, the input is a mixed graph. Here we proved that the integrality gap is
$\Omega(|V|)$ already for $k=2$. In what follows, we try to explain why the linear programming approach works for the first problem while
it fails for the second problem. Moreover, why does the integrality gap example exist for the second problem when $k=2$ and not when $k=1$?

The difference between these problems is explained best when looked at the question of feasibility, which is just a plain orientation problem.
Theorem~\ref{thm:orient} gives a necessary and sufficient
condition when an undirected graph $G$ is $f$-orientable for a
nonnegative crossing $G$-supermodular function.
The characterization includes a condition on every partition and
co-partition of $V$.
 The question whether there exists a
$k$-edge-connected orientation of mixed graph $G=(V,A\cup E)$ has been studied as the problem of
{\em orientation of  mixed graphs}. %
This problem was solved by Frank~\cite{frank96} (see also
\cite[Chapter 16]{frankbook}) in the abstract
framework of orientations covering crossing $G$-supermodular demand functions
(without assuming nonnegativity)\footnote{Negative values of the
  demand function may seem meaningless for the first sight; indeed, we would
  get an equivalent requirement if replacing $f(S)$ by
  $f^+(S)\defeq \max\{f(S),0\}$. However notice that $f^+(S)$ may not
  be crossing supermodular; hence the negative values are needed to
  guarantee the supermodularity property.}. However, the characterization turns
out to be substantially more difficult than in
Theorem~\ref{thm:orient} even for the special case of
$k$-connectivity. The partition and co-partition constraints do not
suffice, but the more general structure of {\em tree-compositions} is
required. We omit the definition here but remark that the fundamental
cuts in the construction in Section~\ref{sec:mixed} form a
tree-composition and the presence of this non-trivial family is what
makes the integrality gap example work.

However, the case $k=1$ is much simpler than $k\ge 2$: Boesch and Tindell
\cite{Boesch80} showed that a mixed graph has a strongly connected
orientation if and only if the underlying undirected graph is
2-edge-connected, and there are no cuts containing only directed arcs
in the same direction. Thus the much simpler family of cut constraints suffices and the tree-composition constraints are not necessary.
This is in accordance with the result that for $k=1$, Khanna, Naor, and
Shepherd \cite{Khanna05} gave a 4-approximation
algorithm, whereas we show a large integrality gap already for $k\ge 2$.

For nonnegative demand functions, we have seen that whereas \eqref{lp:2} is not suitable for iterative rounding, the method can be used for
its projection \eqref{lp} onto  the $x$-space. One might similarly hope that in the mixed setting it is possible to project \eqref{lp:3} to the $x$-space using the characterization in \cite{frank96}, with constraints corresponding to the intricate tree-composition structures instead of partitions and co-partitions only.
Nevertheless, our construction in Section~\ref{sec:mixed} shows that this is not possible, as the $\Omega(|V|)$ integrality gap would also be valid for the projection of \eqref{lp:3}, and therefore iterative rounding, or any other rounding, cannot give a constant approximation. Hence the difficulties arising in the mixed setting are more severe,
where the current approach does not seem to succeed. A natural open question is to obtain a hardness of approximation result matching the integrality gap given in this paper.

We also remark that the results of Frank and Kir\'aly
\cite{frankkiraly} for the minimum cardinality setting are of somewhat
similar flavor. They are able to find the exact optimal solution for
the problem of adding a minimum number of new edges to a graph $G$ so that
it has an orientation covering a nonnegative valued crossing
$G$-supermodular demand function. However, if the demand function can
also take negative values, they only give a characterization of
optimal solutions for the
degree-prescribed  variant of the problem, and the minimum cardinality
setting is left open.

\section*{Acknowledgement}

Mohit Singh would like to thank Seffi Naor and Bruce Shepherd for numerous discussions.

\newpage

\bibliographystyle{abbrv}
\bibliography{conn-orient}

\end{document}